\documentclass[11pt]{article}

\usepackage{amsmath,amssymb,amsthm}
\usepackage{tikz}
\usetikzlibrary{arrows}
\usepackage{palatino}
\usepackage{mathpazo}
\usepackage{stmaryrd}
\usepackage[margin=1in]{geometry}
\usepackage{mathtools}
\mathtoolsset{centercolon}
\usepackage{thm-restate}
\usepackage{hyperref}
\definecolor{darkred}  {rgb}{0.5,0,0}
\definecolor{darkblue} {rgb}{0,0,0.5}
\definecolor{darkgreen}{rgb}{0,0.5,0}
\hypersetup{
  pdftitle = {Everything You Always Wanted to Know About LOCC (But Were Afraid to Ask)},
  pdfauthor = {Eric Chitambar, Debbie Leung, Laura Mancinska, Maris Ozols, Andreas Winter},
  colorlinks = true,
  urlcolor  = blue,         
  linkcolor = darkblue,     
  citecolor = darkgreen,    
  filecolor = darkred       
}

\newtheorem{theorem}{Theorem}
\newtheorem{proposition}{Proposition}
\newtheorem{lemma}{Lemma}
\newtheorem{corollary}{Corollary}
\theoremstyle{definition}

\newtheorem*{remark}{Remark}
\newtheorem*{example}{Example}

\DeclarePairedDelimiter{\norm}{\lVert}{\rVert}

\newcommand{\ket}[1]{|#1\rangle}

\newcommand{\op}[2]{|#1\rangle \langle #2|}

\DeclareMathOperator{\tr}{Tr}
\DeclareMathOperator{\conv}{conv}
\newcommand{\mc}[1]{\mathcal{#1}}
\newcommand{\CP}{\text{CP}}
\newcommand{\dep}{\mathcal{D}}
\DeclareMathOperator{\LOCC}{LOCC}

\DeclareMathOperator{\SEP}{SEP}

\newcommand{\no}[1]{$^{#1}$}

\title{\Large {\bf Everything You Always Wanted to Know About LOCC \protect\\ (But Were Afraid to Ask)}}
\author{
Eric Chitambar,\no{1,2}
Debbie Leung,\no{3}
Laura Man\v{c}inska,\no{3}
Maris Ozols,\no{3,4}
Andreas Winter\no{5,6,7}
\\[4mm]
$^{1}$
\textit{Department of Physics, Southern Illinois University,}\\
\textit{Carbondale, Illinois 62901, USA}\\[2mm]
$^{2}$
\textit{The Perimeter Institute for Theoretical Physics,}\\
\textit{Waterloo, Ontario N2L 2Y5, Canada}\\[2mm]
$^{3}$
\textit{Department of Combinatorics \& Optimization}\\
\textit{and Institute for Quantum Computing, University of Waterloo,}\\
\textit{Waterloo, Ontario N2L 3G1, Canada}\\[2mm]
$^{4}$
\textit{IBM TJ Watson Research Center,}\\
\textit{1101 Kitchawan Road, Yorktown Heights, NY 10598, USA}\\[2mm]
$^{5}$
\textit{ICREA -- Instituci\'{o} Catalana de Recerca i Estudis Avan\c{c}ats,}\\
\textit{Pg. Lluis Companys 23, 08010 Barcelona, Spain}\\
\textit{F\'{\i}sica Te\`{o}rica: Informaci\'{o} i Fenomens Qu\`{a}ntics,}\\
\textit{Universitat Aut\`{o}noma de Barcelona, ES-08193 Bellaterra (Barcelona), Spain}\\[2mm]
$^{6}$
\textit{Department of Mathematics, University of Bristol, Bristol BS8 1TW, U.K.}\\[2mm]
$^{7}$
\textit{Centre for Quantum Technologies, National University of Singapore,}\\
\textit{2 Science Drive 3, Singapore 117542, Singapore}
}

\date{}

\begin{document}
\maketitle

\begin{abstract}
In this paper we study the subset of generalized quantum measurements on finite dimensional systems known as local operations and classical communication (LOCC).  While LOCC emerges as the natural class of operations in many important quantum information tasks, its mathematical structure is complex and difficult to characterize.  Here we provide a precise description of LOCC and related operational classes in terms of quantum instruments.  Our formalism captures both finite round protocols as well as those that utilize an unbounded number of communication rounds.  While the set of LOCC is not topologically closed, we show that finite round LOCC constitutes a compact subset of quantum operations.  Additionally we show the existence of an open ball around the completely depolarizing map that consists entirely of LOCC implementable maps.  Finally, we demonstrate a two-qubit map whose action can be approached arbitrarily close using LOCC, but nevertheless cannot be implemented perfectly.
\end{abstract}

\newpage
\section{Introduction}

The ``distant lab'' paradigm plays a crucial role in both theoretical and experimental aspects of quantum information.  Here, a multipartite quantum system is distributed to various parties, and they are restricted to act locally on their respective subsystems by performing measurements and more general quantum operations.  However in order to enhance their measurement strategies, the parties are free to communicate any classical data, which includes the sharing of randomness and previous measurement results.  Quantum operations implemented in such a manner are known as LOCC (local operations with classical communication), and we can think of LOCC as a special subset of all physically realizable operations on the global system.  This restricted paradigm, motivated by current technological difficulties in communicating quantum data, serves as a tool to study not only quantum correlations and other nonlocal quantum effects, but also resource transformations such as channel capacities.  


Using LOCC operations to study resource transformation is best illustrated in quantum teleportation \cite{Bennett-1993a}.  Two parties, called Alice and Bob, are separated in distant labs.  Equipped with some pre-shared quantum states that characterize their entanglement resource, they are able to transmit quantum states from one location to another using LOCC; specifically, the exchange rate is one quantum bit (qubit) transmitted for one entangled bit (ebit) plus two classical bits (cbits) consumed.  Via teleportation then, LOCC operations become universal in the sense that Alice and Bob can implement \textit{any} physical evolution of their joint system given a sufficient supply of pre-shared entanglement.  Thus entanglement represents a fundamental resource in quantum information theory with LOCC being the class of operations that manipulates and consumes this resource \cite{Bennett-1996b, Nielsen-1999a, Bennett-1999c}.  Indeed, the class of non-entangled or separable quantum states are precisely those that can be generated exclusively by the action of LOCC on pure product states \cite{Werner-1989a}, and any sensible measure of entanglement must satisfy the crucial property that its expected value is non-increasing under LOCC \cite{Bennett-1996a, Vedral-1997a, Horodecki-2000a, Plenio-2007a}.

The intricate structure of LOCC was perhaps first realized over 20 years ago by Peres and Wootters who observed that when some classical random variable is encoded into an ensemble of bipartite product states, the accessible information may be appreciably reduced if the decoders - played by Alice and Bob - are restricted to LOCC operations \cite{Peres-1991a}.  (In fact, results in \cite{Peres-1991a} led to the discovery of teleportation.)  Several years later, Massar and Popescu analytically confirmed the spirit of this conjecture 
by considering pairs of particles that are polarized in the same randomly chosen direction \cite{Massar-1995a}.  It was shown that when Alice and Bob are limited to a finite number of classical communication exchanges, their LOCC ability to identify the polarization angle is strictly less than if they were allowed to make joint measurements on their shared states.  This line of research culminated into the phenomenon of quantum data hiding \cite{Terhal-2001a, DiVincenzo-2002a, Eggeling-2002a}.  Here, some classical data is encoded into a bipartite state such that Alice and Bob have arbitrarily small accessible information when restricted to LOCC, however, the data can be perfectly retrieved when the duo measures in the same lab.

The examples concerning accessible information mentioned above demonstrate that a gap between the LOCC and globally accessible information exists even in the absence of entanglement.  This finding suggests that nonlocality and entanglement are two distinct concepts, with the former being more general than the latter.  Bennett \textit{et al.} were able to sharpen this intuition by constructing a set of orthogonal bipartite pure product states that demonstrate ``nonlocality without entanglement'' in the sense that elements of the set could be perfectly distinguished by so-called separable operations (SEP) but not by LOCC \cite{Bennett-1999a}.  The significance that this result has on the structure of LOCC becomes most evident when considering the Choi-Jamio\l kowski isomorphism between quantum operations and positive operators \cite{Choi-1975a, Jamiolkowski-1972a}.  Separable operations are precisely the class of maps whose Choi matrices are separable, and consequently SEP inherits the relatively well-understood mathematical structure possessed by separable states \cite{Horodecki-1996a, Bruss-2002a}. The fact that SEP and LOCC are distinct classes means that LOCC lacks this nice mathematical characterization, and its structure is therefore much more subtle than SEP.

Thus despite having a fairly intuitive physical description, the class of LOCC is notoriously difficult to characterize mathematically \cite{Bennett-1999a, Donald-2002a}.  Like all quantum operations, an LOCC measurement can be represented by a trace-preserving completely positive map acting on the space of density operators (or by a ``quantum instrument'' as we describe below), and the difficulty is in describing the precise structure of these maps.  Part of the challenge stems from the way in which LOCC operations combine the globally shared classical information at one instance in time with the particular choice of local measurements at a later time.  The potentially unrestricted number of rounds of communication further complicates the analysis.  A thorough definition of finite-round LOCC has been presented in Ref.~\cite{Donald-2002a} thus formalizing the description given in Ref.~\cite{Bennett-1999a}.  

%
Recently, there has been a renewed wave of interest in LOCC alongside new discoveries concerning asymptotic resources in LOCC processing \cite{DeRinaldis-2004a, Chitambar-2011a, Kleinmann-2011a, Childs-2012a}.  It has now been shown that when an unbounded number of communication rounds are allowed, or when a particular task needs only to be accomplished with an arbitrarily small failure rate (but not perfectly), more can be accomplished than in the setting of finite rounds and perfect success rates.  Consequently, we can ask whether a task can be performed by LOCC, and failing that, whether or not it can be approximated by LOCC, and if so, whether a simple recursive procedure suffices.  To make these notions precise, a definition of LOCC and its topological closure is needed.  Here, we aim to extend the formalisms developed in \cite{Bennett-1999a, Donald-2002a, Kleinmann-2011a} so to facilitate an analysis of asymptotic resources and a characterization of the most general LOCC protocols.  Indeed, we hope that this work will provide a type of ``LOCC glossary'' for the research community.

The approach of this article is to describe LOCC in terms of quantum instruments.\footnote{An alternative characterization of LOCC in terms of physical tasks was described to us privately by the authors of Ref.~\cite{Kleinmann-2011a}.  Instead of considering how well an LOCC map approximates a target map (by a distance measure on maps), one can define a success measure for a particular task and study the achievable values via LOCC.  This is particularly useful when the task does not uniquely define a target map.  Here, we define LOCC in terms of quantum instruments, which admits a more precise mathematical description, and a further optimization over possible target maps can be added if one just wishes to focus on success rates.} This will enable us to cleanly introduce the class of infinite-round LOCC protocols as well as the more general class of LOCC-closure.  
%
The explicit definitions for these operational classes are provided in Sect.~\ref{Sect:definitions} as well as a discussion on the relationships among them.  In Sect.~\ref{Sect:topology}, we discuss some topological features possessed by the different LOCC classes.  The main results here involve (i) showing that the set of fixed outcome LOCC protocols possess a non-empty interior and (ii) providing an upper bound on the number of measurements needed per round to implement any finite round, finite outcome LOCC protocol.
In Sect.~\ref{Sect:bipartite} we construct a two-qubit separable instrument that can be approached arbitrarily close using LOCC but nevertheless cannot be implemented perfectly via LOCC.  This finding represents the first of its kind in the bipartite setting.  Finally in Sect.~\ref{Sect:conclusion} we close with a brief summary of results and discuss some additional open problems.  Technical proofs are reserved for the two appendices.

\section{How to define LOCC?}
\label{Sect:definitions}

\subsection{Quantum Instruments}

Throughout this paper, we consider a finite number of finite-dimensional quantum 
systems.  We denote the associated Hilbert space with $\mathcal{H}$ and refer to 
it as the underlying state space of the system.  Let $\mathcal{B}(\mathcal{H})$ be 
the set of bounded linear operators acting on $\mathcal{H}$ and 
$\mathcal{L}(\mathcal{B}(\mathcal{H}))$ the set of all bounded linear maps on 
$\mathcal{B}(\mathcal{H})$.  A (discrete) \textbf{quantum instrument} $\mathfrak{J}$ 
is a family of completely positive (CP) maps $(\mathcal{E}_j : j\in\Theta)$ 
with $\mathcal{E}_j\in\mathcal{L}(\mathcal{B}(\mathcal{H}))$ and $\Theta$ a finite 
or countably infinite index set, such that $\sum_j\mathcal{E}_j$ is 
trace-preserving \cite{Davies-1970a}.  When (w.l.o.g.) $\Theta = \{1,2,\ldots\}$, we write the instrument also
as an ordered list $\mathfrak{J} = (\mathcal{E}_1,\mathcal{E}_2,\ldots)$.  When it is applied to the state $\rho$, $\mathcal{E}_j(\rho)$ represents the 
(unnormalized) postmeasurement state associated with the outcome $j$, which 
occurs with probability ${\rm tr}(\mathcal{E}_j(\rho))$.  We denote the set of instruments with a given index set $\Theta$ as
$\CP[\Theta] \subset \mathcal{L}(\mathcal{B}(\mathcal{H}))^\Theta$.  For $\Theta = \{1,2,\ldots,m\}$ we abbreviate $\CP[\Theta] =: \CP[m]$.  If the index set is unimportant or implicitly clear from the context we often omit it. 

Note that the set of quantum instruments is convex, where addition of 
two instruments and scalar multiplication is defined componentwise. This
obviously requires that the instruments we combine are defined over the
same index set $\Theta$. However, an instrument $\mathfrak{J} \in \CP[\Theta]$
may naturally be viewed as an element in $\Theta' \supset \Theta$ by
padding $\mathfrak{J}$ with zeros, 
i.e.~$\mathcal{E}_j = 0$ for $j\in \Theta'\setminus\Theta$.  Also, observe that for finite $\Theta$ (finite $m$), the set of instruments 
over $\Theta$ is a closed subset of a finite-dimensional vector space.

Given index set $\Theta$, we define a \textbf{quantum-classical} (QC) \textbf{map} over $\Theta$ as a trace-preserving completely positive map (TCP) which sends $\mathcal{B}(\mathcal{H})\to\mathcal{B}(\mathcal{H})\otimes\mathcal{B}(\mathbb{C}^{|\Theta|})$ and is of the form \\$\rho\mapsto\sum_{j\in\Theta}\mathcal{E}_j(\rho)\otimes\op{j}{j}$, where the $\ket{j}$ constitute an orthonormal basis for $\mathbb{C}^{|\Theta|}$.  In this way, we see that for some underlying space $\mathcal{H}$, the set $\CP[\Theta]$ is in a one-to-one correspondence with the set of QC maps over $\Theta$.  For instrument $\mathfrak{J}=(\mathcal{E}_j : j\in\Theta)$, we denote its corresponding QC map by
$\mathcal{E[\mathfrak{J}]}(\cdot) = \sum_{j}\mathcal{E}_j(\cdot) \otimes\op{j}{j}$.


\begin{example}
Consider a POVM measurement $\mathcal{M}$ with Kraus operators $M_1,\dotsc, M_n$ that upon measuring $\rho$ and obtaining outcome $j$ gives postmeasurement state $M^{}_j \rho M_j^\dagger / \tr(M_j^\dagger M^{}_j \rho)$. The instrument corresponding to this POVM is $(\mathcal{E}_1,\dotsc,\mathcal{E}_n)$ where each $\mathcal{E}_j (\rho) = M^{}_j \rho M_j^\dagger $ has only one Kraus operator.
\end{example}

\begin{example}  
Consider a TCP map $\mathcal{N}$ with Kraus operators $M_1,\dotsc, M_n$, i.e., $\mathcal{N}(\rho)=\sum_{i\in[n]} M^{}_i \rho M_i^\dagger$ where $[n] := \{1, 2, \dotsc, n\}$. The instrument corresponding to this TCP map is $(\mathcal{N})$ with only one CP map (which is necessarily trace-preserving).
\end{example}

Let $\mathfrak{J} = (\mathcal{E}_j : j\in\Theta)$ and 
$\mathfrak{J}' = (\mathcal{F}_k : k\in\Theta')$ be quantum instruments. 
We say that $\mathfrak{J}'$ is a \textbf{coarse-graining} of $\mathfrak{J}$
if there exists a partition of the index set $\Theta$, given by
$\Theta = \bigsqcup_{k\in\Theta'} S_k$, such that $\mathcal{F}_k = \sum_{j\in S_k} \mathcal{E}_j$ for each $k\in\Theta'$. Equivalently,
and perhaps more intuitively, one can describe this by the action of a
\textbf{coarse-graining map} $f:\Theta \to \Theta'$, which is 
simply the function $f(k) = j$ for $k \in S_j$ (the sets $S_j$ may be infinite). In this picture we are 
using the coarse-graining map $f$ to post-process the classical information
from the instrument $\mathfrak{J}'$.
Physically, this action corresponds to the discarding of classical information if 
the coarse-graining is non-trivial.  The fully coarse-grained instrument of 
$\mathfrak{J}$ corresponds to the TCP map $\sum_{j}\mathcal{E}_j$, obtainable 
by tracing out the classical register of $\mathcal{E[\mathfrak{J}]}$.  We say 
that an instrument $(\mathcal{E}_j : j\in\Theta)$ is \textbf{fine-grained} if 
each of the $\mathcal{E}_j$ has action of the form $\rho\mapsto M_j \rho M_j^\dagger$ 
for some operator $M_j$.  In this way, the most general instrument can be 
implemented by performing a fine-grained instrument followed by coarse-graining.  



The set of instruments over an index set $\Theta$ carries a metric as follows.
For instruments $\mathfrak{J}=(\mathcal{E}_j : j\in\Theta)$ and 
$\tilde{\mathfrak{J}}=(\tilde{\mathcal{E}}_j : j\in\Theta)$, we use the distance 
measure induced by the diamond norm on the associated QC maps: 
\begin{equation}
\label{Eq:diamond}
  D_{\diamond}(\mathfrak{J},\tilde{\mathfrak{J}})
   := \bigl\| \mathcal{E[\mathfrak{J}]} - \mathcal{E[\tilde{\mathfrak{J}}]} \bigr\|_\diamond
    =\max_{0\leq \rho\leq \mathbb{I}}\sum_{j\in\Theta} \| (\mathcal{I}\otimes\mathcal{E}_j - \mathcal{I}\otimes \tilde{\mathcal{E}}_j)[\rho] \|_1
\end{equation}
where $\norm{\cdot}_\diamond$ is the diamond norm on superoperators \cite{Kitaev-2002a, Watrous-2005a}, $\norm{A}_p$ is the so-called Shatten $p$-norm of $A$, $\mathcal{I}$ is the identity in $\mathcal{L}(\mathcal{B}(\mathcal{H}))$, and $\rho\in\mathcal{B}(\mathcal{H}\otimes\mathcal{H})$.  That this series converges follows from the bound $\norm[\big]{\mathcal{E[\mathfrak{J}]} - \mathcal{E[\tilde{\mathfrak{J}}]} }_\diamond\leq 2$.  To see this observe that $\norm[\big]{\mathcal{E[\mathfrak{J}]} - \mathcal{E[\tilde{\mathfrak{J}}]} }_\diamond$ is equal to 
\[
  \max_{0\leq\rho\leq\mathbb{I}} \norm[\big]{(\mathcal{I}\otimes\mathcal{E[\mathfrak{J}]})[\rho] - (\mathcal{I}\otimes\mathcal{E[\tilde{\mathfrak{J}}]})[\rho]}_1 \leq 
  \max_{0\leq\rho\leq\mathbb{I}} \norm[\big]{(\mathcal{I}\otimes\mathcal{E[\mathfrak{J}]})[\rho]}_1 + \max_{0\leq\rho'\leq\mathbb{I}} \norm[\big]{(\mathcal{I}\otimes\mathcal{E[\mathfrak{J}]})[\rho']}_1.
\]  
As both $\mathcal{E}[\mathfrak{J}]$ and $\mathcal{E}[\mathfrak{J}']$ are trace-preserving, we have the desired uppper bound.
A sequence of instruments $\mathfrak{J}_\nu \in \CP[\Theta]$, $\nu = 1,2,\dotsc$, 
is said to converge to the instrument $\mathfrak{J} = (\mathcal{E}_j : j\in\Theta)$ if 
$\lim_{\nu \rightarrow \infty} D_{\diamond}(\mathfrak{J}_\nu,\mathfrak{J}) = 0$.  

If the index set $\Theta$ is finite, instrument convergence reduces to the pointwise condition: for all $j \in \Theta$, $\lim_{\nu \rightarrow \infty} \| \mathcal{E}_{\nu,j} - \tilde{\mathcal{E}}_j \|_\diamond = 0$. In fact, since $\CP[\Theta]$ is a subset of a 
finite-dimensional real vector space when $\Theta$ is finite, the topology is unique, independent
of the metric, and $\CP[\Theta]$ is complete in the sense that every
Cauchy sequence converges to an element
of $\CP[\Theta]$ (w.r.t. any norm, such as the one described above). All these statements are no longer true for infinite $\Theta$.



\subsection{LOCC Instruments}

For an $N$-partite quantum system, the underlying state space is $\mathcal{H}:=\mathcal{H}^{A_1}\otimes\mathcal{H}^{A_2}\otimes\dotsb\otimes\mathcal{H}^{A_N}$ with $\mathcal{H}^{A_K}$ being the reduced state space of party $K$.  An instrument $\mathfrak{J}^{(K)}=(\mathcal{F}_1,\mathcal{F}_2,\dotsc)$ is called \textbf{one-way local} with respect to party $K$ if each of its CP maps has the form $\mathcal{F}_j = \bigl(\bigotimes_{J \neq K}\mathcal{T}_j^{\smash{(J)}}\bigr) \otimes \mathcal{E}_j^{\,\smash{(K)}}$, where $\mathcal{E}^{\smash{(K)}}$ is a CP map on $\mathcal{B}(\mathcal{H^{\smash{A_K}}})$, and for each $J \neq K$, $\mathcal{T}_j^{\smash{(J)}}$ is some TCP map.  Operationally, this one-way local operation consists of party $K$ applying an instrument $(\mathcal{E}_1,\mathcal{E}_2,\dotsc)$, broadcasting the classical outcome $j$ to all other parties, and party $J$ applying TCP map $\mathcal{T}_j^{\smash{(J)}}$ after receiving this information.

We say that an instrument $\mathfrak{J}'$ is \textbf{LOCC linked} to $\mathfrak{J}=(\mathcal{A}_1,\mathcal{A}_2,\dotsc)$ if there exists a collection of one-way local instruments $\{\mathfrak{J}^{\,\,\smash{(K_j)}}_j =(\mathcal{B}_{1|j},\mathcal{B}_{2|j},\dotsc) : j=1,2,\dotsc \}$ such that $\mathfrak{J}'$ is a coarse-graining of the instrument with CP maps $\mathcal{B}_{j'|j} \circ \mathcal{A}_j$.  Operationally, after the completion of $\mathfrak{J}$, conditioned on the measurement outcome $j$, instrument $\mathfrak{J}^{\,\,\smash{(K_j)}}_j$ is applied followed by coarse-graining (see Fig.~\ref{CoarseGrain}).  Note that we allow the acting party in the conditional instrument $K_j$ to vary according to the previous outcome $j$.


\begin{figure}[t]
  \centering

\begin{tikzpicture}[
  > = latex',
  circ/.style = {circle, draw = black, fill = gray!40,
                 inner sep = 0mm, minimum size = 2mm}]

\def\w{1.6}; 
\def\h{1.1}; 


\node (J) at (0,2*\h) [circ] {};
\node at (3.8*\w,1.65*\h) {$\mathfrak{J} = (\mathcal{A}_1, \mathcal{A}_2)$};

\foreach \a in {1,2} {
  \pgfmathparse{2*\w*(\a-3/2)};
  \let\x = \pgfmathresult;
  \node (A\a) at (\x,\h) [circ] {};
  \draw [->] (J) to node[xshift = 15*(\a-1.5), yshift = 5]
        {$\mathcal{A}_{\a}$} (A\a);
  \foreach \b in {1,2} {
    \node (B\b\a) at ({\x+\w*(\b-3/2)},0) [circ] {};
    \draw [->] (A\a) to node[xshift = 20*(\b-1.5), yshift = 3]
          {$\mathcal{B}_{\b|\a}$} (B\b\a);
  }
}

\path[dashed,<->] (B11) edge[bend right] (B12);
\path[dashed,<->] (B21) edge[bend right] (B22);
\node at (0,-0.8*\h) {(coarse-grain over matching branches)};

\node[align = center] at (3.8*\w,0.2*\h) {Coarse-grained instrument:\\[6pt]%
$\begin{aligned}
 \mathfrak{J}' =
 (&\mathcal{B}_{1|1} \circ \mathcal{A}_1 +
   \mathcal{B}_{1|2} \circ \mathcal{A}_2 , \\
  &\mathcal{B}_{2|1} \circ \mathcal{A}_1 +
   \mathcal{B}_{2|2} \circ \mathcal{A}_2 )
\end{aligned}$%
};

\end{tikzpicture}

  \caption{The instrument $\mathfrak{J}'$ is LOCC linked to the instrument $\mathfrak{J}$.  Conditional instruments $(\mathcal{B}_{1|1},\mathcal{B}_{2|1})$ and $(\mathcal{B}_{1|2},\mathcal{B}_{2|2})$ can be composed with the two elements of $\mathfrak{J}$ so that after coarse-graining, the resulting instrument is $\mathfrak{J}'$.}
  \label{CoarseGrain}
\end{figure}
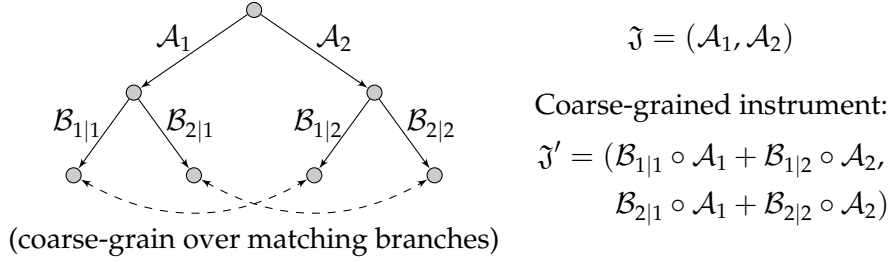

We now define the different classes of LOCC instruments.  For a quantum instrument $\mathfrak{J}\in\CP[\Theta]$, we say that:
\begin{itemize}
\item $\mathfrak{J}\in\text{LOCC}_1$ if $\mathfrak{J}$ is one-way local with 
respect to some party $K$, followed by a coarse-graining map.  
\item $\mathfrak{J}\in\text{LOCC}_{r}$ ($r \geq 2$) if it is LOCC linked to some $\mathfrak{J}\in\text{LOCC}_{r-1}$.
\item $\mathfrak{J}\in\text{LOCC}_{\mathbb{N}}$ if $\mathfrak{J}\in\text{LOCC}_r$ for some $r\in\mathbb{N}=\{ 1,2,\dotsc \}$.
\item $\mathfrak{J}\in\text{LOCC}$ if there exists a sequence $\{\mathfrak{J}_1,\mathfrak{J}_2,\dotsc\}$ in which (i) $\mathfrak{J}_\nu \in\text{LOCC}_{\mathbb{N}}$, (ii) $\mathfrak{J}_\nu$ is LOCC linked to $\mathfrak{J}_{\nu{-}1}$, and (iii) each $\mathfrak{J}_\nu$ has a coarse-graining $\mathfrak{J}_\nu'$ such that $\mathfrak{J}_\nu'$ converges to $\mathfrak{J}$.
\item $\mathfrak{J}\in\overline{\text{LOCC}_{\mathbb{N}}}$ if there exists a sequence $\mathfrak{J}_1,\mathfrak{J}_2,\dotsc$ in which (i) $\mathfrak{J}_j\in\text{LOCC}_{\mathbb{N}}$ and (ii) the sequence converges to $\mathfrak{J}$.
\end{itemize}

Operationally, $\text{LOCC}_{r}$ is the set of all instruments that can be implemented by some $r$-round LOCC protocol.  Here, one round of communication involves one party communicating to all the others, and the sequence of communicating parties can depend on the intermediate measurement outcomes.  The set of instruments that can be implemented by some finite round protocol is then $\text{LOCC}_{\mathbb{N}}$.  On the other hand, so-called infinite round protocols, or those having an unbounded number of non-trivial communication rounds, correspond to instruments in $\text{LOCC}\setminus\text{LOCC}_{\mathbb{N}}$.  The full set of LOCC then consists of both bounded-round protocols as well as the unbounded ones \cite{Bennett-1999a, Kleinmann-2011a, Chitambar-2011a}.  The set $\overline{\text{LOCC}_{\mathbb{N}}}$ is the topological closure of $\text{LOCC}_{\mathbb{N}}$, and from the chain of inclusions $\text{LOCC}_{\mathbb{N}} \subset \text{LOCC}\subset \overline{\text{LOCC}_{\mathbb{N}}}$, we obtain $\overline{\text{LOCC}_{\mathbb{N}}}=\overline{\text{LOCC}}$.



Both $\text{LOCC}$ and $\overline{\text{LOCC}_{\mathbb{N}}}$ consist of all instruments that can be approximated better and better with more LOCC rounds.  The two sets are distinguished by noting that for any instrument in $\text{LOCC}$, its approximation in finite rounds can be made tighter by just continuing for more rounds within a \textit{fixed} LOCC protocol; whereas for instruments in $\overline{\text{LOCC}_{\mathbb{N}}}\setminus \text{LOCC}$, different protocols will be needed for different degrees of approximation. 

Note that according to our definitions, every LOCC instrument is defined with respect to some fixed index set $\Theta$.  However, the instruments implemented during intermediate rounds of a protocol might range over different index sets.  The requirement is that the intermediate instruments can each be coarse-grained into $\Theta$ so to form a convergent sequence of instruments.  This coarse-graining need not correspond to an actual discarding of information.  Indeed, discarding the measurement record midway through the protocol will typically prohibit the parties from completing the final LOCC instrument since the choice of measurement in each round depends on the full measurement history.  On the other hand, often there will be an accumulation of classical data superflous to the task at hand, and the parties will physically perform some sort of coarse-graining (discarding of information), especially at the very end of the protocol.



To complete the picture, we also provide definitions for the related classes of separable (SEP) and positive partial transpose preserving (PPT) instruments.  A multipartite state $\rho^{A_1:A_2:\cdots:A_N}$ is called (fully) separable if it can be expressed as a convex combination of product states with respect to the partition $A_1\mathbin{:}A_2\mathbin{:}\cdots\mathbin{:}A_N$.  Likewise, $\rho^{A_1:A_2:\cdots:A_N}$ is said to have positive partial-transpose (PPT) if the operator obtained by taking a partial transpose with respect to any subset of parties is positive semi-definite.  Now, for the space $\mathcal{H} := \mathcal{H}^{A_1} \otimes \mathcal{H}^{A_2} \otimes \dotsb \otimes \mathcal{H}^{A_N}$, we introduce the auxiliary state space $\mathcal{H}' := \mathcal{H}^{A_1'} \otimes \mathcal{H}^{A_2'} \otimes \dotsb \otimes \mathcal{H}^{A_N'}$ and let $\mathcal{I}'$ denote the identity in $\mathcal{L}(\mathcal{B}(\mathcal{H}'))$.  Let $\mathfrak{J} = (\mathcal{E}_1, \mathcal{E}_2, \dotsc)$ be a quantum instrument acting on $\mathcal{B}(\mathcal{H})$ and consider the state
\begin{equation}
  (\mathcal{I}' \otimes \mathcal{E}_j) [\rho^{A'_1 A^{}_1 : A'_2 A^{}_2 : \cdots : A'_N A^{}_N}].
  \label{Eq:rho}
\end{equation}
We say that
\begin{itemize}
  \item $\mathfrak{J} = (\mathcal{E}_1, \mathcal{E}_2, \dotsc) \in \text{SEP}$ if each $\mathcal{E}_j$ is a separable map \cite{Vedral-1997a,Rains-1997a}, meaning that the state in Eq.~\eqref{Eq:rho} is separable whenever $\rho^{A'_1 A^{}_1 : A'_2 A^{}_2 : \cdots : A'_N A^{}_N}$ is separable;
  \item $\mathfrak{J} = (\mathcal{E}_1, \mathcal{E}_2, \dotsc) \in \text{PPT}$ if each $\mathcal{E}_j$ is a PPT map \cite{Rains-2001a,DiVincenzo-2002a}, meaning that the state in Eq.~\eqref{Eq:rho} is PPT whenever $\rho^{A'_1 A^{}_1 : A'_2 A^{}_2 : \cdots : A'_N A^{}_N}$ is PPT.
\end{itemize}
Operationally, these classes are more powerful than LOCC, yet they are still more 
restrictive than the most general quantum operations.  At the same time, they 
admit a simpler mathematical characterization than LOCC, and this can be used 
to derive many limitations on  LOCC, such as entanglement 
distillation \cite{Rains-1997a, Rains-2001a, Horodecki-1998a} and state 
discrimination \cite{DiVincenzo-2002a, Chefles-2004a}.

Every instrument belonging to $\text{LOCC}_r$, $\text{LOCC}$, $\overline{\text{LOCC}}$, SEP or PPT has an associated index set $\Theta$.  Thus for each $\Theta$, these operational classes naturally become subsets of CP$[\Theta]$, and we denote them as $\text{LOCC}_r[\Theta]$, $\text{LOCC}[\Theta]$, $\overline{\text{LOCC}}[\Theta]$, SEP$[\Theta]$ and PPT$[\Theta]$ respectively.  In the following sections we shall consider these different classes in more detail.


\subsection{Relationships Between the Classes}
\label{subsect:relationship}

For a fixed number of parties $N \geq 2$, and over a sufficiently
large (but universal) index set $\Theta$,
the different LOCC classes are related to each other as follows:
\begin{equation}
\text{LOCC}_1 \subsetneq \text{LOCC}_r \subsetneq \text{LOCC}_{r+1} \subsetneq \text{LOCC}_{\mathbb{N}} \subsetneq \text{LOCC} \subsetneq \overline{\text{LOCC}} \subsetneq \text{SEP} \subsetneq \text{PPT}
\end{equation}
for any $r \geq 2$.  We note that $\text{LOCC}_r \subsetneq \text{LOCC}_{r+k}$ for some $k\in\mathbb{N}$ implies $\text{LOCC}_r \subsetneq \text{LOCC}_{r+1}$.  To see this, suppose that $\text{LOCC}_r=\text{LOCC}_{r+1}$, and let $\mathfrak{J}$ be some instrument in $\text{LOCC}_{r+k}\setminus\text{LOCC}_r$.  Then there exists an implementation of $\mathfrak{J}$ consuming $r+k$ rounds with $\mathfrak{J}_{r+1}$ being the instrument performed during the first $r+1$ rounds of this particular implementation.  But since  $\text{LOCC}_r=\text{LOCC}_{r+1}$, we have $\mathfrak{J}_{r+1}\in\text{LOCC}_r$ and so $\mathfrak{J}\in\text{LOCC}_{r+k-1}$.  Here, we have considered $\Theta$ sufficiently large such that both $\mathfrak{J}$ and $\mathfrak{J}_{r+1}$ are instruments over the same index set (this can always be done by Theorem \ref{Thm:Finite}).  Repeating this argument for $k$ total times gives that $\mathfrak{J}\in\text{LOCC}_{r}$, which is a contradiction.

We now explain why all inclusions are proper.  The operational advantage of $\text{LOCC}_2$ over $\text{LOCC}_1$ is well-known, having been observed in entanglement distillation \cite{Bennett-1996a}, quantum cryptography \cite{Gottesman-2003a}, and state discrimination \cite{Cohen-2007a, Owari-2008a}.  On the other hand, only a few examples have been proven to demonstrate the separation between $\text{LOCC}_r$ and $\text{LOCC}_{r+1}$.  For $N=2$, Xin and Duan have constructed sets containing $O(n^2)$ pure states in two $n$-dimensional systems that require $O(n)$ rounds of LOCC to distinguish perfectly \cite{Xin-2008a}.  For $N \geq 3$, a stronger separation is shown for random distillation of bipartite entanglement from a three-qubit state.  Here the dimension is fixed, and two extra LOCC rounds can always increase the probability of success by a quantifiable amount; moreover, certain distillations only become possible by infinite-round LOCC, thus demonstrating $\text{LOCC}_{\mathbb{N}} \neq \text{LOCC}$ \cite{Chitambar-2011a}.  By studying the same random distillation problem, one can show that $\text{LOCC}\neq \overline{\text{LOCC}}$ \cite{Chitambar-2012a}.  Thus, there are instruments that require different protocols to achieve better and better approximations when more and more LOCC rounds are available, and neither $\text{LOCC}_{\mathbb{N}}$ nor $\text{LOCC}$ is closed.  For $N=2$, we will demonstrate in Sect.~\ref{Sect:bipartite} that $\text{LOCC}_{\mathbb{N}}$ is also not closed for two-qubit systems.  The difference between $\text{LOCC}_{\mathbb{N}}$ and $\text{SEP}$ has emerged in various problems such as state discrimination \cite{Bennett-1999b, Cohen-2007a, Duan-2007a} and entanglement transformations \cite{Chitambar-2008a}.  Proving that $\overline{\text{LOCC}} \neq \text{SEP}$ is more difficult, but it indeed has been demonstrated in Refs.~\cite{Bennett-1999a, Koashi-2007a, Childs-2012a} for the task of state discrimination, as well as Ref.~\cite{Cui-2011a, Chitambar-2012a} for random distillation.  In Sect.~\ref{Sect:bipartite} we will provide another example of this.  Finally, the strict inclusion between PPT and SEP follows from the existence of non-separable states that possess a positive partial transpose \cite{Horodecki-1997a}.

While $\text{LOCC}\subsetneq\text{SEP}$, it is possible to \textit{stochasitcally} perform any instrument from SEP by LOCC.  More precisely, let $\Theta=[m]$ and $\mathfrak{J}=(\mathcal{E}_1,\dotso,\mathcal{E}_m)\in \text{SEP}$ be an arbitrary separable instrument.  Then we say that $\mathfrak{J}$ can be performed by \textbf{Stochastic LOCC} (SLOCC) if there is some nonzero probability $p$ such that the instrument $\mathfrak{J}'=(p\mathcal{E}_1,\dotso p\mathcal{E}_m, (1-p)\mathcal{D})\in \text{CP}[m+1]$ is implementable by LOCC.  Here, $\dep \in\mc{L}(\mathcal{B}(\mathcal{H}))$ is the completely depolarizing TCP map which acts by $\dep(\rho)=\frac{\tr(\rho)}{d}\mathbb{I}_d$, where $\mathbb{I}_d$ is the identity operator in $\mathcal{B}(\mathcal{H})$ and $d$ is the dimension of $\mathcal{H}$.  The outcome $\mathcal{D}$ indicates a failure in implementing $\mathfrak{J}$ and it occurs with probability $1-p$.  It was shown in Ref. \cite{Dur-2000a} that every separable CP map has an SLOCC implementation.  Here we provide a lower bound on the success probability $p$.
\begin{lemma}
\label{Lem:SLOCCprob}
If $\mathfrak{J}=(\mathcal{E}_1,\dotso,\mathcal{E}_m)$ is some $N$-partite separable instrument with $d$ being the total dimension of $\mathcal{H}$.  Then $\mathfrak{J}$ can be implemented by {\upshape SLOCC} with a success probability at least $\frac{1}{md^2}$.
\end{lemma}
\begin{proof}
For all $i\in[m]$ let $\{A_{ij}^1\otimes \dotsc \otimes A_{ij}^N\}_{j\in[d^2]}$ be the set of operators in some Kraus representation of $\mc{E}_i$ with no more than $d^2$ elements. The existence of such a representation can be deduced using Carath\'{e}odory's theorem (see Theorem \ref{Thm:Finite} and its proof). For each Kraus operator $A_{ij}^1\otimes \dotsc \otimes A_{ij}^N$, we choose the individual matrices so that $\norm{A_{ij}^1}_{\infty} = \dotso = \norm{A_{ij}^p}_{\infty}$. This ensures that $\mathbb{I}_{d_k} \geq (A^k_{ij})^\dagger A^k_{ij}$, where $\mathbb{I}_{d_k}$ is the identity acting on subspace $\mathcal{H}^{A_k}$.  Therefore, 
\begin{equation}
  \mc{M}_{ij}^k := \big\{A^k_{ij}, \sqrt{\mathbb{I}_{d_k}-(A^k_{ij})^\dagger A^N_{ij}}\big\}
\end{equation}
is a valid local measurement for each party $k$.  The protocol now consists of the parties first collectively choosing a pair $(i,j)\in[m]\times[d^2]$ uniformly at random.  They then take turns performing their respective local measurements $\mc{M}_{ij}^k$ and broadcasting their result.  If all parties obtain the first outcome, their implementation is a success and they fully coarse-grain the classical data $(i,j)$ over the index $j$ (hence recovering the maps $\mathcal{E}_i$).  If at least one party obtains the second outcome, all the parties locally depolarize and this is a failure outcome.  Coarse-graining over all failure outcomes generates the ($N$-round) LOCC instrument $(\tfrac{1}{md^2}\mathcal{E}_1,\dotso,\tfrac{1}{md^2}\mathcal{E}_m,\tfrac{md^2-1}{md^2}\mathcal{D})$.
\end{proof}

\section{What is the shape of LOCC?}
\label{Sect:topology}

In this section we describe some topological properties of the set $\text{LOCC}$.  

\begin{proposition}
{\upshape LOCC} forms a convex subset of the set of all quantum instruments.
\end{proposition}

\begin{proof}
Given two LOCC instruments $\mathfrak{J}_1$ and $\mathfrak{J}_2$, the convex combination $\lambda\mathfrak{J}_1+(1-\lambda)\mathfrak{J}_2$ can be implemented by LOCC by introducing some globally accessible randomness into the first round of the protocol that determines whether the parties perform $\mathfrak{J}_1$ or $\mathfrak{J}_2$.  
\end{proof}

Next we turn to the question of LOCC interior.  In what follows, let $d$ denote the total dimension of the $N$-party state space $\mathcal{H}=\mathcal{H}^{A_1}\otimes\dotso\otimes\mathcal{H}^{A_N}$.
\begin{theorem}
\label{Thm:LOCCinterior}
For a finite index set $\Theta=[m]$, {\upshape LOCC} has a non-empty interior.
\end{theorem}
\noindent To prove this, we will show the existence of a non-empty ball that consists entirely of LOCC instruments around the completely depolarizing instrument $\mathfrak{D}:=\left( \frac{1}{m} \mathcal{D}, \dotsc, \frac{1}{m} \mathcal{D} \right)$.\footnote{Via private communication, we have learned that Marco Piani has independently obtained a similar result for TCP maps.}
Our argument will rely on the following result of Gurvits and Barnum.
\begin{proposition}[\cite{Barnum03}]
\label{Prop:SepBall}
If $A\in \mc{B}(\mathcal{H})$ and $\norm{A}_2\leq R_N:=2^{1-N/2}$, then $\mathbb{I}_d+A$ is a separable operator with respect to the partition $A_1:A_2:\dotsb:A_N$.
\end{proposition}
\noindent While this is a statement about separable operators, we can easily translate it into a statement about separable maps.  Recall that for a CP map $\mathcal{E}\in\mathcal{L}(\mathcal{B}(\mathcal{H}^{A_1}\otimes\dotso\otimes\mathcal{H}^{A_N}))$, with $d_i$ being the dimension of $\mathcal{H}^{A_i}$, the $N$-partite Choi matrix is given by 
\[\Omega^{A_1'A^{}_1:\dotsb:A_N'A^{}_N}_\mathcal{E} := (\mathcal{I}^{A_1A_2\dotso A_N}\otimes\mathcal{E}^{A_1A_2\dotso A_N}) [\Phi^{A_1'A^{}_1} \otimes\dotso\otimes \Phi^{A_N'A^{}_N}],\]
where $\Phi^{A'_iA^{}_i}=\sum_{i,j=1}^{d_i}\op{ii}{jj}$ \cite{Choi-1975a}.  In short, the Choi matrix is obtained by distributing maximally entangled states across two copies of the original $N$-partite system and applying the map $\mathcal{E}$ to just half of it.  It is known that $\mathcal{E}$ is a separable CP map if and only if $\Omega^{A_1'A^{}_1:\dotsb:A_N'A^{}_N}$ is a separable operator with respect to the partition $A_1'A^{}_1:\dotsb:A_N'A^{}_N$ \cite{Cirac-2001a}.  Then as a first corollary to Proposition \ref{Prop:SepBall}, we have: 
\begin{corollary}
\label{Cor:SEPball}
Let $\mathfrak{J} = (\mc{E}_1, \dotsc, \mc{E}_m)\in \text{{\upshape CP}}[m]$. Then the instrument
\begin{equation}
\label{Eq:DepMix}
  (1-\eta) \mathfrak{D} + 
  \eta \mathfrak{I}
 \in \SEP,
\end{equation}
for all $\eta \leq R_{\SEP}:= R_N / (md^2 + R_N) = (md^2 2^{\frac{N}{2}-1} + 1)^{-1}$.
\end{corollary}
\begin{proof}
By the above discussion, we need to show that for each CP map in Eq.~\eqref{Eq:DepMix}, its corresponding Choi matrix is separable.  For the depolarizing map $\mathcal{D}$, we have $\Omega_\mathcal{D}=\frac{1}{d}\mathbb{I}_{d^2}$, and therefore, the Choi matrix for $(1-\eta)\frac{1}{m}\mathcal{D}+\eta\mathcal{E}_i$ is
$(1-\eta) \tfrac{1}{md}\mathbb{I}_{d^2} + \eta \Omega_{\mc{E}_i}$.  From Proposition~\ref{Prop:SepBall} we get that $\mathbb{I}_{d^2} + \frac{\eta m d}{1-\eta} \Omega_{\mc{E}_i}$ is separable whenever $\frac{\eta m d}{1-\eta} \norm{\Omega_{\mc{E}_i}}_2 \leq R_N$.  Note that $\norm{\Omega_{\mc{E}_i}}_2 \leq \norm{\Omega_{\mc{E}_i}}_1  = d$. Therefore, choosing $\eta \leq \frac{R_N}{md^2+R_N}$ yields the desired statement.
\end{proof}
\noindent Now we are in a position to take Proposition~\ref{Prop:SepBall} one step further and prove Theorem \ref{Thm:LOCCinterior}.  
\begin{corollary}
\label{Cor:LOCCBall}
Any instrument $\mathfrak{J} = (\mc{E}_1, \dotsc, \mc{E}_m)\in\text{{\upshape CP}}[m]$ such that 
\begin{equation}
  D_\diamond(\mathfrak{D},\mathfrak{J}) \leq R_{\LOCC}
\end{equation}
can be implemented by an $N$-round LOCC protocol, where $R_{\LOCC} := \frac{R_{\SEP}}{m^2d^4}
 =(m^3d^6  2^{\frac{N}{2}-1}+ m^2d^4)^{-1}$.
\end{corollary}   
\begin{proof}
We begin by decomposing each $\mathcal{E}_i$ in $\mathfrak{J}$ as $\mathcal{E}_i=(1-md\delta)\frac{1}{m}\mathcal{D}+md\delta\mathcal{F}_i$ where $\mathcal{F}_i=\frac{1}{m}\mathcal{D}-\frac{1}{md\delta}(\frac{1}{m}\mathcal{D}-\mathcal{E}_i)$.  By considering the Choi matrices $\Omega_{\mathcal{F}_i}$, we see that each $\mathcal{F}_i$ is CP when taking $\delta=\max_{i\in[m]}\norm{\frac{1}{md}\mathbb{I}_{d^2}-\Omega_{\mathcal{E}_i}}_\infty$.  From the definition of the diamond norm, it immediately follows that
\begin{align}
  D_\diamond(\mathfrak{D},\mathfrak{J}) &= 
   \max_{0\leq \rho\leq \mathbb{I}} \sum_{i\in[m]}\norm[\Big]{(\mathcal{I}\otimes\frac{1}{m}\mathcal{D} -
      \mathcal{I}\otimes \mathcal{E}_i)[\rho]}_1 
    \geq \sum_{i\in[m]} \frac{1}{d} 
      \norm[\Big]{\frac{1}{md}\mathbb{I}_{d^2} - \Omega_{\mathcal{E}_i}}_1
\notag\\
  & \geq \sum_{i\in[m]}\frac{1}{d} \norm[\Big]{\frac{1}{md}
    \mathbb{I}_{d^2}-\Omega_{\mathcal{E}_i}}_\infty
  \geq \frac{1}{d}\max_{i\in[m]} \norm[\Big]{\frac{1}{md}\mathbb{I}_{d^2} 
    - \Omega_{\mathcal{E}_i}}_\infty
  = \frac{\delta}{d}.  
\end{align}
Therefore, we have $\mathfrak{J}=(1-md\delta)\mathfrak{D}+md\delta\mathfrak{J}'$ where $\mathfrak{J}'=(\mathcal{F}_1,\dotso,\mathcal{F}_m)\in\text{CP}[m]$ and $md\delta \leq md^2D_\diamond(\mathfrak{D},\mathfrak{J})$.  Now Corollary \ref{Cor:SEPball} guarantees that $(1-\eta)\mathfrak{D}+\eta\mathfrak{J}'\in\text{SEP}$ when $\eta\leq R_{\SEP}$.  From Lemma \ref{Lem:SLOCCprob} and the convexity of LOCC, there exists an $N$-round LOCC protocol that implements this instrument with probability $p$ and depolarizes (fails) with probability $1-p$, for any $p\leq \frac{1}{md^2}$.  By coarse-graining the failure outcome into each of the $m$ success outcomes by equal amounts, we have that the instrument 
\begin{equation}
\label{Eq:LOCCMix}
p\bigl((1-\eta)\mathfrak{D}+\eta\mathfrak{J}'\bigr)+(1-p)\mathfrak{D}=(1-p\eta)\mathfrak{D}+p\eta\mathfrak{J}'\in \LOCC_N
\end{equation}
whenever $p\eta\leq R_{\SEP}/(md^2)$.  Taking $p\eta=md\delta\leq md^2 D_\diamond(\mathfrak{D},\mathfrak{J})$, it follows that $\mathfrak{J}=(1-md\delta)\mathfrak{D}+md\delta\mathfrak{J}'\in\text{LOCC}_N$ whenever $D_\diamond(\mathfrak{D},\mathfrak{J})\leq  R_{\SEP}/(md^2)^2$.
\end{proof}

We next turn to the question of compactness.  As LOCC itself is not closed (see Sect.~\ref{subsect:relationship}), clearly it is not a compact set.  However, we will prove that when restricted to finite round protocols with finite number of outcomes, compactness indeed holds.  Such a result might not be entirely obvious; it is conceivable that, before coarse-graining to a finite number of outcomes in the final round, the protocol requires intermediate measurements with an unbounded number of outcomes.  We will show that no such requirement can exist for a finite round LOCC protocol.  Our first step is to apply Carath\'{e}odory's Theorem~\cite{Rockafellar-1996a} repeatedly to bound the number of measurement outcomes in each step.  The following theorem is proven in Appendix~\ref{appendixCaratheodory}.

\begin{restatable}{theorem}{ThmFinite}\label{Thm:Finite}
For an $N$-partite system, suppose that $\mathfrak{J}=(\mathcal{E}_1,\dotsc,\mathcal{E}_m)$ is an instrument in {\upshape $\text{LOCC}_r$}.  Let $d_K$ denote the local dimension of party $A_K$ and $D=\prod_{K=1}^Nd_K$ the dimension of the underlying global state space.  Then there exists an $r$-round protocol that implements $(\mathcal{E}_1,\dotsc,\mathcal{E}_m)$ such that each instrument in round $l \in [r]$ consists of no more than $mD^{4(r-l+1)}$ {\upshape CP} maps of the form $M(\cdot)M^\dagger$ for some $D \times D$ matrix $M$, with some overall coarse-graining performed at the end of the protocol.
\end{restatable}


\begin{corollary}
Consider an $N$-partite system of total dimension $D$.  For any $r,m\in\mathbb{N}$, the subset of instruments in $\text{{\upshape LOCC}}_r$ with $m$ outcomes is compact in the set of all quantum instruments.  
\end{corollary}

\begin{proof}
By the previous theorem, any such LOCC instrument belonging to this set can be characterized by matrices:
\begin{align}
\label{Eq:ops}
\bigl\{
 M^{}_{i_1}, M_{i_2}^{(i_1)}, \dotsc, M_{i_r}^{(i_1\dotsc i_{r-1})} :
 i_1 \in [n_1],
 i_2 \in [n_2], \dotsc,
 i_r \in [n_r]
\bigr\}
\end{align}
where (i) for each $i_l$ the $M_{i_l}^{(i_1\dotsc i_{l-1})}$ are square matrices of some fixed size $<D$, (ii) $n_l \leq mD^{4(r-l-1)}$ for $1\leq l \leq r$ and (iii) $\sum_{i_l=1}^{n_l} \bigl(M_{i_l}^{(\vec{s})}\bigr)^\dagger M_{i_l}^{(\vec{s})} = \mathbb{1}$ for all $1\leq l \leq r$ and outcome strings $\vec{s}$.  As this is a finite collection of algebraic constraints, the set of feasible instruements is both closed and bounded.
\end{proof}

\section{Is LOCC closed?}
\label{Sect:bipartite}

Bipartite entanglement is known to behave differently than its multipartite counterpart in many circumstances.  One such example is the LOCC transformation of some pure state from one form to another.  While in the bipartite case such a transformation can always be completed in one round of LOCC if it is possible at all \cite{Lo-1997a}, the same is not true for tripartite pure state manipulations.  

Among the separation results listed in Sect.~\ref{subsect:relationship}, whether $\text{LOCC}=\overline{\text{LOCC}}$ in the bipartite case is unknown prior to this work.  In this section, we exhibit a bipartite instrument $\mathfrak{J}$ which is in $\overline{\text{LOCC}}$ but not in $\text{LOCC}$.  This instrument acts only on two qubits, and is given by
\begin{equation}
  \mathfrak{J} = (\mathcal{E}_{00}, \mathcal{E}_{01}, \mathcal{E}_{10})
  \label{Eq:J}
\end{equation}
where
\begin{align}
  \mathcal{E}_{00}(\rho)
  &:= \op{11}{11} \rho \op{11}{11}, \notag \\
  \mathcal{E}_{01}(\rho)
  &:= \sum_{i=1}^2 \bigl(T_i\otimes\op{0}{0}\bigr) \rho
                   \bigl(T^\dagger_i\otimes\op{0}{0}\bigr), \label{Eq:maps} \\
  \mathcal{E}_{10}(\rho)
  &:= \sum_{i=1}^2 \bigl(\op{0}{0}\otimes T_i\bigr) \rho
                   \bigl(\op{0}{0}\otimes T^\dagger_i\bigr), \notag
\end{align}
with
\begin{align}
  \label{Eq:Ti}
  T_1 &:=
  \begin{pmatrix}
    \frac{1}{\sqrt{3}} & 0 \\
    0 & \frac{1}{\sqrt{3}}
  \end{pmatrix}, &
  T_2 &:=
  \begin{pmatrix}
    \frac{1}{\sqrt{6}} & 0 \\
   0 & \sqrt{\frac{2}{3}}
  \end{pmatrix}.
\end{align}

We construct in Sect.~\ref{Sect:limitLOCC} a sequence of LOCC instruments that converges to $\mathfrak{J}$, thereby showing $\mathfrak{J} \in \overline{\text{LOCC}}$.  In Sect.~\ref{Sect:LOCCimpossible} we prove that $\mathfrak{J} \notin \text{LOCC}$, by demonstrating that the induced instrument $\mathfrak{J} \otimes \mathbb{I}$ on three qubits, where $\mathbb{I}$ denotes the identity instrument, is not in $\text{LOCC}$.  We show that, if $\mathfrak{J} \otimes \mathbb{I}$ is in $\text{LOCC}$, its effect on a particular three-qubit pure state violates the monotonicity of a quantity that we derive in Sect.~\ref{Sect:randdist}.  This quantity pertains to a random entanglement distillation task that we also review and discuss.

\subsection{Proof of \texorpdfstring{$\mathfrak{J} \in \overline{\text{LOCC}}$}{}}
\label{Sect:limitLOCC}
We construct a sequence of LOCC instruments $\{\mathfrak{J}_1,\mathfrak{J}_2,\dotsc\}$ taken directly from the Fortescue-Lo random distillation scheme \cite{Fortescue-2007a}, except here we just consider bipartite systems.  Let $\epsilon >0$ be fixed.  Consider the measurement ${\cal M}(\rho) = M_0 \, \rho \, M_0 \otimes \op{0}{0} + M_1 \, \rho \, M_1 \otimes \op{1}{1}$ where   
\begin{align}
M_0&:=\sqrt{1-\epsilon}\op{0}{0}+\op{1}{1},&
M_1&:=\sqrt{\epsilon}\op{0}{0}
\end{align}
are diagonal.  For each $\nu\in\mathbb{N}$, $\mathfrak{J}_\nu$ is implemented by the following protocol.  Alice and Bob each perform the measurement ${\cal M}$ locally and share the measurement outcome.  If the joint outcome is one of $01$, $10$, or $11$, they stop.  If the joint outcome is $00$, then they repeat the same measurement procedure again.  After a maximum of $\nu$ iterations they stop.  Then, coarse-graining is applied to obtain 
$\mathfrak{J}_\nu = (\mathcal{E}_{\nu00},\mathcal{E}_{\nu01},\mathcal{E}_{\nu10},\mathcal{E}_{\nu11})$ where $\mathcal{E}_{\nu ij}$ includes all the cases when Alice and Bob stop upon obtaining the joint outcome $ij$.  More specifically, these four CP maps are respectively generated by the following sets of Kraus operators: 
$\{ M_0^{\nu}\otimes M_0^{\nu} \}$, 
$\{M_0^\mu\otimes M_1^{} M_0^{\mu-1} : \mu \in [\nu]\}$, 
$\{M_1^{} M_0^{\mu-1}\otimes M_0^{\mu} : \mu \in [\nu]\}$, and 
$\{M_1^{} M_0^{\mu-1}\otimes M_1^{} M_0^{\mu-1} : \mu \in [\nu]\}$.


To see what this instrument looks like, we consider the Choi matrices of its CP maps.  
The four Choi matrices corresponding to the four elements in $\mathfrak{J}_{\nu}$ (up to the ordering of spaces in the tensor product) are:
\begin{align}
  \Omega_{\nu 00} \; &= \sum_{w,x,y,z=0}^1 \op{wx}{yz}^{A'B'} \otimes
  (M_0^{\nu}\otimes M_0^{\nu}) \op{wx}{yz}^{AB} (M_0^{\nu}\otimes M_0^{\nu}), \notag \\
  \Omega_{\nu 01} \; &= \sum_{w,x,y,z=0}^1 \op{wx}{yz}^{A'B'} \otimes \sum_{\mu=1}^{\nu}
  (M_0^\mu\otimes M_1M_0^{\mu-1}) \op{wx}{yz}^{AB} (M_0^\mu\otimes M_0^{\mu-1}M_1), \notag \\
  \Omega_{\nu 10} \; &= \sum_{w,x,y,z=0}^1 \op{wx}{yz}^{A'B'} \otimes \sum_{\mu=1}^{\nu}
  (M_1M_0^{\mu-1}\otimes M_0^\mu) \op{wx}{yz}^{AB} (M_0^{\mu-1}M_1\otimes M_0^\mu), \notag \\
  \Omega_{\nu 11} \; &= \sum_{w,x,y,z=0}^1 \op{wx}{yz}^{A'B'} \otimes \sum_{\mu=1}^{\nu}
  (M_1M_0^{\mu-1}\otimes M_1M_0^{\mu-1}) \op{wx}{yz}^{AB} (M_0^{\mu-1}M_1\otimes M_0^{\mu-1}M_1).
\end{align}

We want to show that $\lim_{\epsilon \rightarrow 0} \lim_{\nu \rightarrow \infty} \mathfrak{J}_{\nu} = \mathfrak{J}$ where the target instrument has been padded to become $\mathfrak{J} = (\mathcal{E}_{00}, \mathcal{E}_{01}, \mathcal{E}_{10}, 0)$ where the CP maps $\mathcal{E}_{ij}$ are defined in Eq.~\eqref{Eq:maps}.  We will do this by considering the limiting matrices $\Omega_{ij}:= \lim_{\epsilon \rightarrow 0} \lim_{\nu \rightarrow \infty} \Omega_{\nu ij}$.
Note that formally we have to choose $\epsilon$ to be some function of $\nu$ for specifying a sequence of instruments $\mathfrak{J}_1,\mathfrak{J}_2,\dotsc$ that converges to $\mathfrak{J}$.  However, we have verified that in all cases the double limit agrees with the single limit $\nu \to \infty$ if we choose $\epsilon := \nu^{-c}$ for some $0 < c < 1$. Thus, for the sake of simplicity of the argument we will compute the double limit.

For $ij=11$, note that $M_1^{} M_0^{\mu-1} = M_0^{\mu-1} M_1^{} = \sqrt{\epsilon(1-\epsilon)^{\mu-1}} \op{0}{0}$, so we have
\begin{equation}
\Omega_{\nu 11} = \sum_{\mu=1}^{\nu} \epsilon^2 (1-\epsilon)^{2\mu-2}\op{00}{00}^{A'A}\otimes\op{00}{00}^{B'B} \,.
\end{equation}
In the $\nu \to \infty$ limit we get $\epsilon^2 \sum_{\mu = 0}^{\infty} (1 - \epsilon)^{2\mu} = \frac{\epsilon^2}{1 - (1-\epsilon)^2} = \frac{\epsilon}{2-\epsilon}$, so $\Omega_{11} = \lim_{\epsilon \to 0} \lim_{\nu \to \infty} \Omega_{\nu 11} = 0$ as desired. Similarly, for $ij=00$, it is easy to see that $\Omega_{00} = \op{11}{11}^{A'A} \otimes \op{11}{11}^{B'B}$.  

For $ij=01$, note that $M_0^{\mu} = \sqrt{(1-\epsilon)^{\mu}} \op{0}{0} + \op{1}{1}$, so we have
\begin{align}
  \Omega_{\nu 01}
  &= \sum_{\mu=1}^\nu
     \Biggl( \sum_{w,y=0}^1 \op{w}{y}^{A'} \otimes M_0^\mu \,\op{w}{y}^{A} M_0^\mu \Biggr) \! \otimes \!
     \Biggl( \sum_{x,z=0}^1 \op{x}{z}^{B'} \otimes M_1 M_0^{\mu-1} \,\op{x}{z}^{B} M_0^{\mu-1} M_1 \Biggr) \\
  &= \sum_{\mu=1}^{\nu}
     \begin{pmatrix*}[l]
             (1-\epsilon)^\mu  & \sqrt{(1-\epsilon)^{\mu}} \\
       \sqrt{(1-\epsilon)^\mu} & 1
     \end{pmatrix*}^{A'A} \otimes \epsilon (1-\epsilon)^{\mu-1} \op{00}{00}^{B'B} \\
  &= \frac{\epsilon}{1-\epsilon} \sum_{\mu=1}^{\nu}
     \begin{pmatrix*}[l]
       (1-\epsilon)^{2\mu}   & (1-\epsilon)^{3\mu/2} \\
       (1-\epsilon)^{3\mu/2} & (1-\epsilon)^{\mu}
     \end{pmatrix*}^{A'A} \otimes \op{00}{00}^{B'B}
\end{align}
where the first tensor factor in the last two lines is written in the basis $\{\ket{00}_{A'A},\ket{11}_{A'A}\}$.  Evaluating the geometric series as $\nu \to \infty$ and then $\epsilon \to 0$, we see that
\begin{align}
\Omega_{01}&=\begin{pmatrix}1/2&2/3\\2/3&1\end{pmatrix}^{A'A}\otimes\op{00}{00}^{B'B}.
\end{align}
By permuting the parties, we obtain $\Omega_{10}$. 

Finally, it is easy to verify that $\Omega_{01}$ is indeed the Choi matrix of the map $\mathcal{E}_{01}$ given by Eq.~\eqref{Eq:maps}, and similarly for all other $\Omega_{ij}$.  Since the Hilbert space dimension is $4$ (a small constant), entry-wise convergence of the Choi matrix is equivalence to the convergence of instruments defined in Eq.~\eqref{Eq:diamond}.  

\subsection{Digression: Random Concurrence Distillation}
\label{Sect:randdist}

As mentioned in Sect.~\ref{Sect:bipartite}, we will prove LOCC infeasibility of the bipartite instrument $\mathfrak{J}$ in Eq.~\eqref{Eq:J} by showing infeasibility of the induced tripartite transformation $\mathfrak{J} \otimes \mathbb{I}$. Specifically, this is shown by considering the action of the instrument $\mathfrak{J} \otimes \mathbb{I}$ on the state $\ket{W}=\sqrt{1/3}(\ket{100}+\ket{010}+\ket{001})$.  

The state $\ket{W}$ is the canonical representative of the so-called W-class of states which consists of all three-qubit pure states that are locally unitarily (LU) equivalent to $\ket{\vec{x}}=\sqrt{x_0}\ket{000}+\sqrt{x_A}\ket{100}+\sqrt{x_B}\ket{010}+\sqrt{x_C}\ket{001}$ \cite{Dur-2000a} for some $\vec{x}=(x_A,x_B,x_C)$ (with $x_0=1-x_A-x_B-x_C$).  
For any normalized pure W-class state $\ket{\psi}$, we can unambiguously represent it by an $\vec{x}_\psi$ \cite{Kintas-2010a} as follows.  First, it is not difficult to show that, if $\ket{\vec{x}}$ is LU equivalent to $\ket{\vec{x}'}$, and all entries of $\vec{x}$ are nonzero, then $\vec{x}=\vec{x}'$. If $x_K=0$ for exactly one party $K$, signifying that only bipartite entanglement exists, then by Schmidt decomposing the two-qubit state of the entangled parties, we can take $x_0=0$ and can choose $x_A\geq \max\{x_B, x_C\}$ if $x_A \neq 0$, and $x_B \geq x_C$ otherwise.  Finally, if $x_K = 0$ for two or more parties, the state is fully separable, and we set $x_A=1$.  
%
Thus, LOCC transformations on W-class states can be viewed as transformations between their unique representations, and can be understood by tracking the changes to $\vec{x}$.  

Now, consider the following transformation task.  Suppose that Alice, Bob and Charlie are in possession of some W-class state, and they wish to obtain a bipartite pure state held by either Charlie-Alice or Charlie-Bob such that the expected \textit{concurrence} is maximized.  We call this task random concurrence distillation.  Recall that for a two-qubit pure state of the form $\ket{\psi}=\alpha\ket{01}+\beta\ket{10}$, its concurrence $C(\psi)$ is given by $2|\alpha\beta|$ \cite{Wootters-1998a}.  

The described problem generalizes two different tasks in quantum information processing: \textit{fixed-pair concurrence distillation} and \textit{random EPR combing}.  In fixed-pair concurrence distillation, one target pair of parties is \textit{a priori} specified, and the goal is for the third party to measure in such a way that maximizes the average post-measurement concurrence shared between the target pair.  This optimal value is known as the Concurrence of Assistance (COA).  In Ref.~\cite{Laustsen-2003a} it is shown that for a W-class state with $x_0=0$, its COA is $2\sqrt{x_Jx_K}$ when parties $J$ and $K$ constitute the target pair.  In contrast, in random EPR combing, two target pairs of parties are \textit{a priori} specified, and the goal is to maximize the probability that either of these pairs obtains an EPR state.  For example, if Alice is the common party to both target pairs, it has recently been shown that for W-class states with $x_0=0$, the optimal probability is $2(x_B+x_C-x_Bx_C/x_A)$ when $x_A\geq \max\{x_B,x_C\}$ and $2x_A$ otherwise \cite{Chitambar-2012a}.  Random concurrence distillation is a hybrid of these two problems in that the goal is to randomly distill entanglement to two different pairs, but instead of using the probability of obtaining an EPR state as a success measure, we consider the average concurrence distilled.  The following lemma, proved in Appendix~\ref{appendixB}, provides an upper bound of the random concurrence distillable between a fixed party $\star\in\{A,B,C\}$ and the remaining two parties.  

\begin{restatable}{lemma}{LemMonotone}\label{Lem:Concurrence Monotone}
Let $\star\in\{A,B,C\}$ denote any fixed party.  For a W-class state $\sqrt{x_A}\ket{100}+\sqrt{x_B}\ket{010}+\sqrt{x_C}\ket{001}$ take $\{n_1,n_2\}=\{A,B,C\}\setminus\{\star\}$ such that $x_{n_1}\geq x_{n_2}$, and consider the function
\begin{equation}
\label{Eq:Concfunc}
\mathcal{C}(\vec{x}) :=
\begin{cases}
  2 \sqrt{x_\star x_{n_1}} + \frac{2}{3} x_{n_2} \sqrt{\frac{x_\star}{x_{n_1}}}
    & \text{if $x_{n_1} \neq 0$,} \\
  0 & \text{otherwise.}
\end{cases}
\end{equation}
The function $\mathcal{C}$ is non-increasing on average under LOCC.  Furthermore, it is strictly decreasing on average when either party ``$\star$'' or $n_1$ performs a non-trivial measurement.
\end{restatable}

\begin{remark}
Note that if $x_\star = 0$, then $x_{n_1}>0$ and so $\mathcal{C}(\vec{x}) = 0$ by Eq.~\eqref{Eq:Concfunc}.  If $x_{n_1}=0$, then $\mathcal{C}(\vec{x})$ is defined to be $0$ since $x_{n_2}=0$ and so the state is fully separable.  When $x_{n_2}=0$, the function $\mathcal{C}$ reduces to the bipartite concurrence measure between parties $n_1$ and $\star$.  Consequently, what Lemma~\ref{Lem:Concurrence Monotone} says is that for any tripartite to bipartite entanglement conversion of a W-class state, the maximum average concurrence that some fixed party ``$\star$'' can share with any of the other two parties is no greater than $\mathcal{C}(\vec{x})$.  Finally, observe that for the state $\ket{W}=\sqrt{1/3}(\ket{100}+\ket{010}+\ket{001})$, we have $\mathcal{C}(W)=8/9$.
\end{remark}


\subsection{LOCC Impossibility}
\label{Sect:LOCCimpossible}

We now put the pieces together to prove that $\mathfrak{J} \notin \text{LOCC}$ where $\mathfrak{J}$ is the bipartite instrument defined in Eq.~\eqref{Eq:J}.  Suppose $\mathfrak{J} \otimes \mathbb{I}$ is applied to $\ket{W}=\sqrt{1/3}(\ket{100}+\ket{010}+\ket{001})$.  This induces the following state transformation:
\begin{equation}
\label{Eq:statetrans}
\op{W}{W} \mapsto
\begin{cases}
 \omega^{(AC)}\otimes\op{0}{0}^{(B)} & \text{with probability $1/2$},\\
 \op{0}{0}^{(A)}\otimes\omega^{(BC)} & \text{with probability $1/2$},
\end{cases}
\end{equation}
where 
\begin{equation*}
\omega:=\begin{pmatrix}0&0&0&0\\0&1/3&4/9&0\\0&4/9&2/3&0\\0&0&0&0\end{pmatrix}.
\end{equation*}
Since Eq.~\eqref{Eq:statetrans} describes a mixed state transformation, Lemma~\ref{Lem:Concurrence Monotone} cannot be used directly.  However, using the convex roof construction \cite{Horodecki-2001a}, we can extend the monotone $\mathcal{C}$ to mixed W-class states.  Recall that the set of mixed W-class states consists of all convex combinations of pure W-class states, biseparable states, and product states \cite{Acin-2001b}.  It represents a convex compact set in state space that is closed under LOCC operations.  Therefore, the following is well-defined.
\begin{corollary}
Let $\star\in\{A,B,C\}$ denote any fixed party and suppose $\rho$ is some mixed W-class state.  Define the function 
\begin{equation}
\hat{\mathcal{C}}(\rho)=\min_{p_i,\ket{\phi_i}}\sum_{p_i} p_i \, \mathcal{C}(\phi_i)
\end{equation}
where the minimization is taken over all pure state decompositions of $\rho=\sum_ip_i\op{\phi_i}{\phi_i}$, and $\mathcal{C}$ is the pure state entanglement monotone defined by Eq.~\eqref{Eq:Concfunc}.  Then $\hat{\mathcal{C}}$ is an entanglement monotone on the class of mixed W-class states.
\end{corollary}

For transformation \eqref{Eq:statetrans}, we see that $\hat{\mathcal{C}}(W)=8/9$ and the final value of $\hat{\mathcal{C}}$ in both outcomes reduces to $\mathcal{C}(\omega)$, the concurrence of $\omega$ (i.e., it is the convex-roof extension of the bipartite pure state concurrence).  Using Wootters' formula for the concurrence \cite{Wootters-1998a}, we can compute that $\mathcal{C}(\omega)=8/9$.  Thus, $\hat{\mathcal{C}}$ must remain invariant on average during each measurement in a protocol that performs transformation \eqref{Eq:statetrans}.  However, we can always decompose the first non-unitary measurement into a local pure state transformation on $\ket{W}$, which by Lemma~\ref{Lem:Concurrence Monotone} will cause $\mathcal{C}$ to strictly decrease (choose $n_1$ as the first measuring party).  And because $\hat{\mathcal{C}}$ reduces to $\mathcal{C}$ on pure states, we therefore have that $\hat{\mathcal{C}}$ will necessarily decrease on average during the first non-trivial measurement.  Hence, transformation \eqref{Eq:statetrans} cannot be performed by LOCC, and so likewise, the original bipartite instrument is infeasible by LOCC.

This argument can also be used to prove the LOCC impossibility of other separable bipartite instruments.  For instance, consider the three CP maps $\mathcal{E}_i(\rho) := \Pi_i^{} \rho \Pi_i^\dagger$ where
\begin{align}
\Pi_1 &= \Bigl({\tfrac{1}{\sqrt{2}}}\op{0}{0}+\op{1}{1}\Bigr) \otimes \op{0}{0},\notag\\
\Pi_2 &= \op{0}{0} \otimes \Bigl(\tfrac{1}{\sqrt{2}}\op{0}{0}+\op{1}{1}\Bigr),\\
\Pi_3 &= \op{1}{1} \otimes \op{1}{1}.\notag
\end{align}
As $\sum_{i=1}^3 \Pi_i^\dagger \Pi^{}_i = I$ we have that $(\mathcal{E}_1,\mathcal{E}_2,\mathcal{E}_3)$ is a separable quantum instrument, and its action on $\ket{W}$ will yield bipartite mixed states with an average concurrence of $2\sqrt{2}/3>8/9$.  Hence, this transformation cannot even be approximated asymptotically, which provides another example of the operational gap between $\overline{\text{LOCC}}$ and SEP.

\section{What did we learn?}
\label{Sect:conclusion}

In this article, we have closely studied the structure of LOCC operations.  In light of recent findings concerning the nature of asymptotic LOCC processes, we have adopted the formalism of quantum instruments to precisely characterize the topological closure of LOCC.  Additionally, we have proven that the class of LOCC instruments acting on two qubits is not closed.  This resolves an open problem and reveals the complexity of LOCC even when dealing with the smallest bipartite systems.

There are a few interesting questions related to our work that deserve additional investigation.  First, all known examples that separate $\text{LOCC}$ from $\overline{\text{LOCC}}$ or $\overline{\text{LOCC}}$ from SEP make use of the classical information obtained from a quantum measurement.  For quantum channels with no classical register, it is unknown whether the same separation results hold (although, $\text{LOCC}_\mathbb{N}$ can be separated from SEP by such channels \cite{Chitambar-2008a}).  For instance, if we coarse grain over the three different maps given by Eq.~\eqref{Eq:maps}, is the resulting channel feasible by LOCC?  We conjecture that it is not, however our current proof techniques are unable to show this.  

On the other hand, one can ask how the operational classes compare if one is only interested in the classical information extracted from a quantum measurement, i.e., if attention is restricted only to POVMs.  While the state discrimination results take such an approach to separate $\overline{\text{LOCC}}$ from SEP, the random distillation examples demonstrating $\text{LOCC}_\mathbb{N} \neq \text{LOCC} \neq \overline{\text{LOCC}}$ depend crucially on the quantum outputs of the measurement.  Thus, it may be possible that $\text{LOCC}=\overline{\text{LOCC}}$ for POVMs.  We draw the reader's attention to Ref.~\cite{Kleinmann-2011a} in which a particular state discrimination problem is presented that directly questions this possibility. 

Finally, by using the distance measure between instruments described in Sect.~\ref{Sect:definitions}, one can meaningfully inquire about the size in separation between operational classes.  When given some instrument in SEP, what is the closet LOCC instrument?  Furthermore, is this distance related to the nonlocal resources needed to implement the separable instrument?   We hope this paper stimulates further research into such questions concerning the structure of LOCC.

\section*{Acknowledgments}
We thank Matthias Kleinmann, Hermann Kampermann, and Dagmar Bru\ss{} for a helpful discussion on asymptotic LOCC.  Additionally, much appreciation is extended to Daniel Gottesman for spotting an error in an earlier version of Sect.~\ref{Sect:bipartite}. DL, LM, and MO were supported by CFI, ORF, CIFAR, CRC, NSERC Discovery grant, NSERC QuantumWorks grant, MITACS, the Ontario Ministry of Research and Innovation, and the US ARO/DTO.
MO acknowledges additional support from the DARPA QUEST program under contract number HR0011-09-C-0047.
AW was supported by the European Commission (STREP ``QCS'' and
Integrated Project ``QESSENCE''), the ERC (Advanced Grant ``IRQUAT''),
a Royal Society Wolfson Merit Award and a Philip Leverhulme Prize.
The Centre for Quantum Technologies is funded by the Singapore
Ministry of Education and the National Research Foundation as part
of the Research Centres of Excellence programme.

\section*{Appendices}

\appendix

\section{Proof of Theorem~\ref{Thm:Finite}}
\label{appendixCaratheodory}

\ThmFinite*

\begin{proof}
We will prove Theorem~\ref{Thm:Finite} in multiple steps. A general $r$-round LOCC instrument can be represented by a tree partitioned into $r$ levels.  Within each level are nodes that correspond to the different one-way local LOCC instruments performed in that round.  The nodes are specified by their respective measurement histories $(i_1 i_2 \dotsc i_l)$.  Many parties may perform a quantum operation at each node, but only one party can perform a non-trace-preserving operation, and this party may vary across different nodes at each level.

Our first task is to convert a general protocol into one for which (i) except for the final round, the dimensions of the input and output spaces are the same for each map, and (ii) the party acting non-trivially at any node along the same level is fixed.  To obtain (i) we first fine-grain each local map into the form specified by the Theorem, i.e., $\rho\mapsto M\rho M^\dagger$.  This modification will only increase the number of edges coming out from a node, but will not change the total number of rounds, and the original instrument can be recovered by suitable coarse-graining at the end.  Then, we apply the polar decomposition $M=UA$ where $A$ is a square matrix and $U$ is an isometry.  Thus whenever $M(\cdot)M^\dagger$ is performed within the protocol, it can be replaced with $A(\cdot)A^\dagger$ combined with a pre-application of $U$ in the next level.  In other words, if the same party applies Kraus operator $N$ next in the original protocol, then $NU$ is applied instead.  Polar decomposing $NU$, we obtain $U'A'$ where $A'$ again maps the initial system to itself, and $U'$ is some other isometry to be moved to yet the next level.  Doing this inductively for all levels yields condition (i) in which the state lives in the same input system throughout, except for the last step. 

To describe simplification (ii), we introduce some terminology.  For an $r$-round LOCC protocol, we say its \textbf{measurement-ordered expansion} is a protocol obtained by increasing the number of rounds through the addition of trivial maps so that there is a unique predetermined party who can act non-trivially along all nodes in a given level.  This new protocol will consist of no more than $Nr$ levels.  A \textbf{measurement-ordered compression} is a reversal of the expansion.  Specifically, for a given LOCC protocol, consider every node $(i_1 \dotsc i_{l-2})$ in which a trivial map is performed at that node (this will be a measurement in round $l-1$), and suppose that $\mathfrak{T}$ is the subsequent one-way local instrument implemented in round $l$.  We modify the original protocol by performing $\mathfrak{T}$ instead of the trivial map at node $(i_1 \dotsc i_{l-2})$.  Doing this for each trivial operation and compressing recursively generates the measurement-ordered compression of the original protocol.

Now we return to the specific instrument $\mathfrak{J}=(\mathcal{E}_1,\dotsc,\mathcal{E}_m)$.  For simplicity we will assume that $N=2$, and more general cases follow by analogous constructions.  Suppose we are given some measurement-ordered protocol that can be compressed to $r$ rounds and which implements $\mathfrak{J}$.  Again for simplicity we will assume that the protocol just consists of $4$ rounds in which the order of measurement is Alice, Bob, Alice, Bob; the case of more rounds can be proved by induction.  For this protocol, the CP maps performed at node $(i_1 \dotsc i_{l-1})$ will be denoted by $\mathcal{A}^{(i_1 \dotsc i_{l-1})}_{i_1}$ ($\mathcal{B}^{(i_1 \dotsc i_{l-1})}_{i_1}$) when Alice (Bob) is the acting party in round $l$.  Thus the entire instrument $\mathfrak{J}$ can be expressed through the TCP map
\begin{equation}
  \mathcal{E}[\mathfrak{J}](\rho)
  = \sum_{\lambda=1}^m \mathcal{E}_\lambda(\rho) \otimes \op{\lambda}{\lambda}
  = \sum_{\lambda=1}^m \sum_{\vec{i}}
    \delta_{\lambda}(\vec{i})
    \mathcal{B}^{(i_1i_2i_{3})}_{i_4} \circ
    \mathcal{A}^{(i_1i_2)}_{i_3} \circ
    \mathcal{B}^{(i_1)}_{i_2} \circ
    \mathcal{A}_{i_1}(\rho) \otimes \op{\lambda}{\lambda},
\end{equation}
where $\delta_\lambda(\vec{i})=1$ if $\vec{i}=(i_1,i_2,i_3,i_4)$ is included in the coarse-graining of $\mathcal{E}_\lambda$ and $0$ otherwise.  Our goal is to obtain a bound on the number of different $\vec{i}$ needed to implement $\mathfrak{J}$.  We construct an operator representation of $\mathcal{E}[\mathfrak{J}]$ by introducing a copy of the local Hilbert space for each measurement performed in each round.  Then, we have the operator
\begin{align}
  \label{Eq:Choistate}
  \Omega&^{A_1'A_1^{}:\cdots
          :B_4'B_4^{}}
  = \sum_{\lambda=1}^m \sum_{\vec{i}} \delta_{\lambda}(\vec{i})
    (\mathbb{I}^{A'_1} \otimes \mathcal{A}_{i_1}) [\Phi^{A'_1A_1^{}}]
    \otimes \dotsb \otimes
    (\mathbb{I}^{B'_4} \otimes \mathcal{B}^{(i_1 \dotsc i_3)}_{i_4}) [\Phi^{B'_{4}B_{4}^{}}]
    \otimes \op{\lambda}{\lambda}.
\end{align}
Here, $\Phi^{A_i'A_i^{}}=\sum_{l,m=1}^{d_i}\op{ll}{mm}$ is a maximally entangled state on two copies of the space possessed by $A_i$ (likewise for $\Phi^{B'_iB_i^{}}$), and the $\mathcal{A}^{(i_1 \dotsc i_{l-1})}_{i_l}$ and $\mathcal{B}^{(i_1 \dotsc i_{l-1})}_{i_l}$ act on the ``unprimed'' copies in their respective rounds.  The action of $\mathcal{E}[\mathfrak{J}]$ is recovered by first supposing that the input state $\rho$ is shared across local registers $A_0$ and $B_0$.  Then $\mathcal{E}[\mathfrak{J}](\rho)$ is given by the following equality (see Fig.~\ref{fig:Protocol}):
\begin{equation}
\label{Eq:Choiequality}
  \mathcal{E}[\mathfrak{J}](\rho)
  = \tr_{A_0^{}B_0^{}A_1^{}A_1'B_2^{}B_2'A'_3B'_4}
    \Bigl[
      \bigl( \rho^{A_0B_0} \otimes \Omega^{A_1^{}A_1':\cdots:B_4^{}B_4'} \bigr)
      \bigl( \Phi^{A_0^{}A'_1}\Phi^{B_0^{}B'_2}\Phi^{A_1^{}A'_3}\Phi^{B_2^{}B'_4} \bigr)
    \Bigr].
\end{equation}

\begin{figure}
  \centering

\def\h{0.9cm} 
\def\W{6.6cm} 
\def\w{1.2cm} 

\def\d{0.2cm} 
\def\l{0.4cm} 

\begin{tikzpicture}

  \foreach \i/\P in {1/A, 2/B, 3/A, 4/B} {
    \draw (0,{-\h*(2*\i  )}) to (\W,{-\h*(2*\i  )});
    \draw (0,{-\h*(2*\i+1)}) to (\W,{-\h*(2*\i+1)});
    \draw (0,{-\h*(2*\i  )}) -- (-\d,{-\h*(2*\i+1/2)}) -- (0,{-\h*(2*\i+1)});
    \draw (-\l,{-\h*(2*\i  )}) node {$\P'_{\i}$};
    \draw (-\l,{-\h*(2*\i+1)}) node {$\P_{\i}$};
    \path (\w*\i,{-\h*(2*\i+1)}) coordinate (\P\i);
  }

  \draw[double, double distance = 1.5pt] (A1) -- (B2) -- (A3) -- (B4);

  \foreach \i/\P/\tt in {1/A/{}, 2/B/(i_1), 3/A/(i_1 i_2), 4/B/(i_1 i_2 i_3)} {
    \draw (\P\i) node[draw, fill = white] {$\mathcal{\P}_{i_{\i}}^{\tt}$};
  }

  \foreach \i in {1,2} {
    \draw (\W,{-\h*(2*\i+1)}) -- (\W+2*\d,{-\h*(2*\i+2.5)}) -- (\W,{-\h*(2*\i+4)});
  }

  \draw [dashed] (\W-\w/2,-1.8*\h) -- (\W-\w/2,-9.2*\h);
  \draw (\W-\w/2,-2*\h+\l) node {$\Omega$};

  \def\c{1.5cm}

  \foreach \i/\P in {1/A,2/B} {
    \draw[dotted] (-\c,{-\h*(\i-1)}) --
                  (\W,{-\h*(\i-1)}) --
                  (\W+2*\d,{-\h*(\i+(\i-1)*0.5)}) --
                  (\W,{-\h*2*\i});
    \draw (-\c-\l,{-\h*(\i-1)}) node {$\P_0$};
  }
  \draw (-\c-2.5*\l,-0.5*\h) node {$\rho \; \Bigg\{$};

  \foreach \i in {1,2} {
    \draw (\W,{-\h*(2*\i+5)}) -- (\W+\c,{-\h*(2*\i+5)});
  }

\end{tikzpicture}

  \caption{Construction of $\Omega$ with four rounds of LOCC according to Eq.~\eqref{Eq:Choistate}.  From the state $\Omega$ the action of $\mathcal{E}[\mathfrak{J}]$ on $\rho$ can be obtained by inputting $\rho$ into wires $A_0 B_0$ and applying the projector $\ket{\Phi}$ onto each of the the pairs $A_0^{}A_1'$, $B_0^{}B_2'$, $A_1^{}A_3'$, and $B_2^{}B_4'$ as in Eq.~\eqref{Eq:Choiequality}.}
  \label{fig:Protocol}
\end{figure}
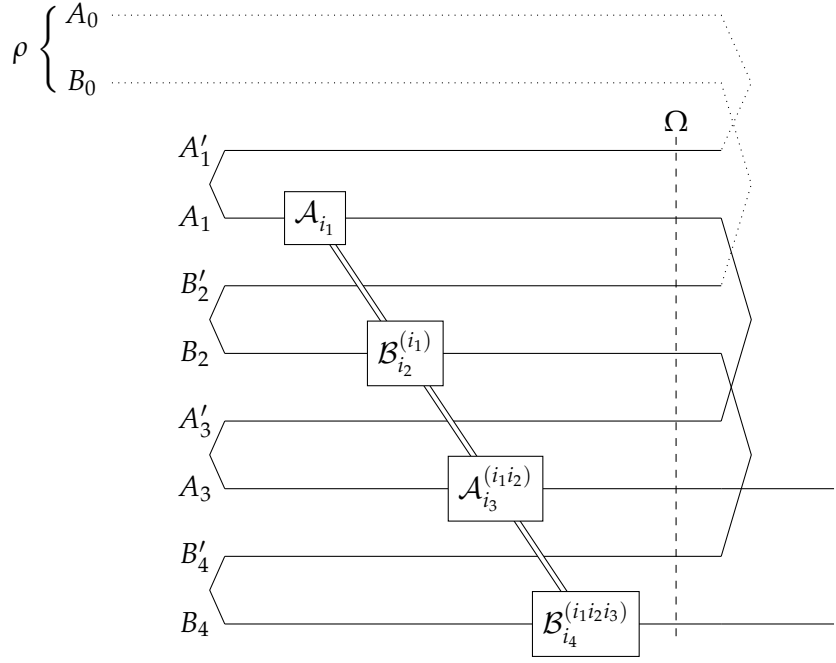

We next use Carath\'{e}odory's Theorem to rewrite the sum over $\vec{i}$ in Eq.~\eqref{Eq:Choistate}.
\begin{lemma}[Carath\'{e}odory's Theorem~\cite{Rockafellar-1996a}]
Let $S$ be a subset of $\mathbb{R}^n$ and $\conv(S)$ its convex hull.  Then any $x \in \conv(S)$ can be expressed as a convex combination of at most $n+1$ elements of $S$.
\end{lemma}
\noindent As a Hermitian operator, every (unnormalized) mixed quantum state can be expressed by $(\dim)^2$ real numbers, where $\dim$ is the dimension of the underlying space on which the operator acts.  We first consider the sum over $\lambda$ and $i_{4}$ in Eq. \eqref{Eq:Choistate} with $i_1, i_2, i_3$ fixed.  This sum can be expressed as
\begin{align}
\label{Eq:Carath1}
\rho^{(i_1i_2i_3)}
&:=\sum_{\lambda=1}^m\sum_{i_4}\delta_\lambda(\vec{i})
(\mathbb{I}^{B'_4}\otimes\mathcal{B}^{(i_1i_2i_3)}_{i_4})[\Phi^{B_4'B_4^{}}]
\otimes\op{\lambda}{\lambda}\notag\\
&=\sum_{\lambda=1}^m\sum_{i_4}q_{\lambda, i_{4}}^{(i_1i_2i_3)}(\mathbb{I}^{B'_4}\otimes\mathcal{B}^{(i_1i_2i_3)}_{i_4})[\Phi^{B_4'B_4^{}}]
\otimes\op{\lambda}{\lambda}
\end{align}
where (i)~for each $\vec{i}$, $q_{\lambda, i_4}^{(i_1i_2i_3)}\geq 0$ is nonzero for only one value of $\lambda$ (the one for which $\delta_\lambda(\vec{i}) = 1$), (ii)~for each $\lambda$, $q_{\lambda, i_4}^{(i_1i_2i_3)}$ is nonzero for at most $d_B^4$ values of $\vec{i}$ (by Carath\'{e}odory's Theorem), and (iii)~$\sum_{\lambda=1}^m\sum_{i_4}q_{\lambda, i_4}^{(i_1i_2i_3)}\mathcal{B}^{(i_1i_2i_3)}_{i_4}$ is trace-preserving.  Thus by (i) and (ii) the sum in Eq. \eqref{Eq:Carath1} contains no more than $md_B^4$ terms.  

From these operators $\rho^{(i_1i_2i_3)}$, we construct the set 
\begin{equation}
T^{(i_1i_2)}:=\bigl\{(\mathbb{I}^{A'_3}\otimes\mathcal{A}^{(i_1i_2)}_{i_3})[\Phi^{A_3'A_3^{}}]\otimes\rho^{(i_1i_2i_3)}:i_3\bigr\}
\end{equation} and its convex hull $\conv(T^{(i_1i_2)})$.  For each fixed pair $i_1i_2$, we can apply Carath\'{e}odory's Theorem again when summing over $i_3$ in Eq. \eqref{Eq:Choistate} to obtain
\begin{align}
\sum_{i_3}\sum_{\lambda=1}^m\sum_{i_4}&\delta_\lambda(\vec{i})
(\mathbb{I}^{A'_3}\otimes\mathcal{A}^{(i_1i_2)}_{i_3})[\Phi^{A_3'A_3^{}}]\otimes
(\mathbb{I}^{B'_4}\otimes\mathcal{B}^{(i_1i_2i_3)}_{i_4})[\Phi^{B_4'B_4^{}}]\otimes\op{\lambda}{\lambda}\notag\\
&=\sum_{i_3}
(\mathbb{I}^{A'_3}\otimes\mathcal{A}^{(i_1i_2)}_{i_3})[\Phi^{A_3'A_3^{}}]\otimes\rho^{(i_1i_2i_3)}\notag\\
&=\sum_{i_3}q_{i_3}^{(i_1i_2)}
(\mathbb{I}^{A'_3}\otimes\mathcal{A}^{(i_1i_2)}_{i_3})[\Phi^{A_3'A_3^{}}]\otimes\rho^{(i_1i_2i_3)}\notag\\
&=\sum_{i_3, i_4}\sum_{\lambda=1}^mq_{i_3}^{(i_1i_2)}q_{ \lambda, i_4}^{(i_1i_2i_3)}
(\mathbb{I}^{A'_3}\otimes\mathcal{A}^{(i_1i_2)}_{i_3})
[\Phi^{A_3'A_3^{}}]\otimes
(\mathbb{I}^{B'_4}\otimes\mathcal{B}^{(i_1i_2i_3)}_{i_4})
[\Phi^{B_4'B_4^{}}]\otimes\op{\lambda}{\lambda},
\end{align}
where $q_{i_3}^{(i_1i_2)}\geq 0$ is nonzero for at most $m(d_Ad_B)^4$ values of $i_3$ and $\sum_{i_3}q_{i_3}^{(i_1i_2)}\mathcal{A}^{(i_1i_2)}_{i_3}$ is trace-preserving.  Repeating this argument on the sums over $i_1$ and $i_2$ respectively, Eq.~\eqref{Eq:Choistate} becomes
\begin{equation*}
\Omega^{A_1'A_1^{}:\cdots:B_4'B_4^{}}=
\sum_{\lambda=1}^m\sum_{\vec{i}}q_{i_1}^{}q_{i_2}^{(i_1)}q_{i_3}^{(i_1i_2)} q_{\lambda, i_4}^{(i_1i_2i_3)}
(\mathbb{I}^{A'_1}\otimes\mathcal{A}_{i_1})[\Phi^{A_1'A_1^{}}]
\otimes\cdots\otimes
(\mathbb{I}^{B'_4}\otimes\mathcal{B}^{(i_1i_2i_3)}_{i_4})[\Phi^{B_4'B_4^{}}]
\otimes\op{\lambda}{\lambda}
\end{equation*}
where again, $q_{i_2}^{(i_1)}\geq 0$ is nonzero for at most $m(d_Bd_Ad_B)^4$ values of $i_2$ while $q_{i_1}\geq 0$ is nonzero for at most $m(d_Ad_Bd_Ad_B)^4$ values of $i_1$.  Both $\sum_{i_1}q_{i_1}\mathcal{A}_{i_1}$ and $\sum_{i_2}q_{i_2}^{(i_1)}\mathcal{B}_{i_2}^{(i_1)}$ are trace-preserving.  

Thus the modified protocol consists of the instrument sequence
\begin{enumerate}
  \item Alice:    $\{q_{ i_1}\mathcal{A}_{i_1}: q_{i_1}>0\}$,
  \item Bob: \,\, $\{q_{ i_2}^{(i_1)}\mathcal{B}^{(i_1)}_{i_2}:q_{ i_2}^{(i_1)}>0\}$,
  \item Alice:    $\{q_{ i_3}^{(i_1i_2)}\mathcal{A}^{(i_1i_2)}_{i_3}:q_{ i_3}^{(i_1i_2)}>0\}$,
  \item Bob: \,\, $\{q_{\lambda, i_4}^{(i_1i_2i_3)}\mathcal{B}^{(i_1i_2i_3)}_{i_4}:q_{\lambda, i_4}^{(i_1i_2i_3)}>0\}$.
\end{enumerate}
An overall coarse-graining for each value of $\lambda$ is applied after the final measurement to obtain $\mathfrak{J}$.  Finally, observe that this modified protocol only differs from the original in the one-way local instruments performed at each node.  Furthermore, if a trivial operation is performed at any node in the original protocol, it will likewise be performed at the same node in this modified protocol.  Therefore, the ability to compress the original protocol to $r$ rounds of LOCC implies that the modified protocol can also be compressed to $r$ rounds.
\end{proof}

\section{Proof of Lemma~\ref{Lem:Concurrence Monotone}}
\label{appendixB}

\LemMonotone*

\begin{proof}
Note that if any of the $x_A, x_B, x_C$ is zero, the lemma is already true, so, we can focus on the case when all three parameters are nonzero.  The proof is analogous to the ones given in Refs.~\cite{Chitambar-2012a, Cui-2011a} where similar W-class entanglement monotones were derived.  We can decompose each local measurement into a sequence of binary measurements \cite{Anderson-2008a, Oreshkov-2005a}.  
By continuity of $\mathcal{C}$, it is then sufficient to show that $\mathcal{C}$ is non-increasing on average under such a measurement.  Each such measurement is specified by two Kraus operators $\{ M_\lambda : \lambda = 1,2 \}$, which can always be expressed in the form
\begin{equation}
  M_\lambda := U_\lambda \cdot
  \begin{pmatrix}
        \sqrt{a_\lambda} & b_\lambda \\
    0 & \sqrt{c_\lambda}
  \end{pmatrix},
\end{equation}
where $a_\lambda, c_\lambda \geq 0$, $b_\lambda \in \mathbb{C}$ and $U_\lambda$ is unitary (which acts trivially on the representation $\vec{x}$).  By the completeness relation $\sum_{\lambda=1}^2 M_\lambda^\dagger M_\lambda^{} = I$ we have
\begin{equation}
  \label{Eq:Completeness}
  \sum_{\lambda=1}^2
    \begin{pmatrix}
      a_\lambda & \sqrt{a_\lambda} b_\lambda \\
      \sqrt{a_\lambda} b^*_\lambda & |b_\lambda|^2 + c_\lambda
    \end{pmatrix}
  = \begin{pmatrix}
      1&0\\0&1
    \end{pmatrix},
\end{equation}
thus $a_1 + a_2 = 1$ and $c_1 + c_2 \leq 1$ with equality if and only if $b_1 = b_2 = 0$.


The state change induced by any such measurement can be understood readily by considering a concrete example, say when the first party measures:  
\begin{multline}
  \label{Eq:Alice measures}
  \sqrt{x_0}\ket{000}+\sqrt{x_A}\ket{100}+\sqrt{x_B}\ket{010}+\sqrt{x_C}\ket{001} \\
  \mapsto
  \frac{1}{\sqrt{p_\lambda}}
  \bigl[(\sqrt{a_\lambda x_0} + b_\lambda \sqrt{x_A}) \ket{000} +
         \sqrt{c_\lambda x_A} \ket{100} +
         \sqrt{a_\lambda x_B} \ket{010} +
         \sqrt{a_\lambda x_C} \ket{001}\bigr],
\end{multline}
where the final state is normalized with $p_\lambda$ being the probability of obtaining outcome $\lambda$.  The post-measurement state in Eq.~\eqref{Eq:Alice measures} is represented by vector $\vec{x}_\lambda = (\frac{c_\lambda}{p_\lambda} x_A, \frac{a_\lambda}{p_\lambda} x_B, \frac{a_\lambda}{p_\lambda} x_C)$.  More generally, when party $K \in \{A,B,C\}$ measures and outcome $\lambda$ occurs, 
then the representing vector components transform as 
\begin{equation}
  x_K \mapsto \frac{c_\lambda}{p_\lambda}x_K \quad \text{and} \quad
  x_J \mapsto \frac{a_\lambda}{p_\lambda}x_J \quad
  \text{for $J \in \{A,B,C\} \setminus \{K\}$}.
\label{Eq:statechange}
\end{equation}


We will consider three separate cases: (i)~when $K = \star$, (ii)~when $K \in \{n_1, n_2\}$ with $x_{n_1} > x_{n_2}$, and (iii)~when $K \in \{n_1, n_2\}$ and $x_{n_1} = x_{n_2}$.  Note from Eq.~\eqref{Eq:statechange} that in case (i) the assignments of parties $n_1$ and $n_2$ will remain unchanged in the pre and post-measurement states.  To simplify analysis, we can impose this also in case (ii) by assuming a sufficiently weak measurement (i.e., one for which $a_1 \approx c_1$ and $a_2 \approx c_2$). This is without loss of generality since any measurement can be decomposed as a sequence of weak measurements. However, in case (iii), the assignment of parties $n_1$ and $n_2$ in the post-measurement states can differ from the pre-measurement state and might even depend on the outcome.

Define the average change in $\mathcal{C}$ incurred by the measurement as
\begin{equation}
\overline{\Delta\mathcal{C}}:= p_1\mathcal{C}(\vec{x}_1)+ p_2\mathcal{C}(\vec{x}_2)-\mathcal{C}(\vec{x}).
\end{equation}
where $p_\lambda$ is the probability to obtain outcome $\lambda$ and $\vec{x}_\lambda$ is the representation of the corresponding post-measurement state.

\newcommand{\Case}[1]{\textbf{Case (#1)}}

\Case{i}
If party ``$\star$'' makes a measurement, 
\begin{equation}
\overline{\Delta\mathcal{C}}=
\left(2\sqrt{x_\star x_{n_1}}+\frac{2}{3}x_{n_2}\sqrt{\frac{x_\star}{x_{n_1}}}\right)
\bigl(\sqrt{c_1a_1}+\sqrt{c_2a_2}-1\bigr) \leq 0,
\end{equation}
with equality obtained only under the trivial measurement $a_1 = c_1 = d$ and $a_2 = c_2 = 1-d$ for some $d \in [0,1]$.
%

\Case{ii}
Suppose either party $n_1$ or $n_2$ measures and $x_{n_1} > x_{n_2}$.  If party $n_2$ measures, the assignment of parties in all post-measurement states remains the same if $a_\lambda x_{n_1} > c_\lambda x_{n_2}$ for all $\lambda \in \{1,2\}$.  Subject to constraints imposed by Eq.~\eqref{Eq:Completeness}, we can guarantee this by assuming that the measurement is sufficiently weak (i.e., $a_1 \approx c_1$ and $a_2 \approx c_2$).  Then
\begin{align}
  \overline{\Delta\mathcal{C}}
 &= \sum_{\lambda=1}^2 \left( a_\lambda 2 \sqrt{x_\star x_{n_1}} + c_\lambda \frac{2}{3} x_{n_2}
    \sqrt{\frac{x_\star}{x_{n_1}}} \right) 
  - \biggl( 2\sqrt{x_\star x_{n_1}}+\frac{2}{3}x_{n_2}\sqrt{\frac{x_\star}{x_{n_1}}} \biggr) \notag \\
 &= (c_1 + c_2 - 1) \; \frac{2}{3} \; x_{n_2} \sqrt{\frac{x_\star}{x_{n_1}}} \leq 0,
\end{align}
with equality attained only when $c_1 + c_2 = 1$ (and thus $b_1 = b_2 = 0$).
Similarly, if party $n_1$ performs a sufficiently weak measurement (i.e., $c_\lambda x_{n_1} > a_\lambda x_{n_2}$ for all $\lambda \in \{1,2\}$),
\begin{align}
\label{Eq:Function change}
  \overline{\Delta\mathcal{C}}
 &= 2 \sqrt{x_\star x_{n_1}} \bigl(\sqrt{a_1c_1}+\sqrt{a_2c_2}-1\bigr)
  + \frac{2}{3}x_{n_2}\sqrt{\frac{x_\star}{x_{n_1}}}
    \biggl(
      a_1 \sqrt{\frac{a_1}{c_1}} +
      a_2 \sqrt{\frac{a_2}{c_2}} - 1
    \biggr).
\end{align}
Moreover, under the weakness constraint,
\begin{equation}
  \frac{\partial}{\partial c_2} \overline{\Delta\mathcal{C}}
= \sqrt{x_\star x_{n_1}} \sqrt{\frac{a_2}{c_2}}
  \left( 1- \frac{1}{3} \frac{x_{n_2}}{x_{n_1}} \frac{a_2}{c_2} \right) > 0.
\end{equation}
Thus, for fixed $a_1$ (thus also $a_2$) and $c_1$, the value of $\overline{\Delta\mathcal{C}}$ is maximized when $c_2=1-c_1$.  Substituting this into Eq.~\eqref{Eq:Function change} and evaluating around the point $(a_1, c_1) = (1/2, 1/2)$, we see that to second order in $(a_1-1/2)$ and $(c_1-1/2)$:
\begin{equation}
  \overline{\Delta\mathcal{C}} \approx \sqrt{x_\star x_{n_1}}
  \left( \frac{x_{n_2}}{x_{n_1}} - 1 \right) (a_1-c_1)^2 \leq 0,
\end{equation}
with equality obtained again only by the trivial measurement.  Thus, $\overline{\Delta\mathcal{C}}$ is non-positive in some neighborhood of $(a_1,c_1) = (1/2,1/2)$.

\Case{iii}
Finally, consider when $x_{n_1} = x_{n_2}$ and either party $n_1$ or $n_2$ measures. Parties $n_1$ and $n_2$ are symmetric prior to measurement, but their ordering can possibly change across the two post-measurement states.  Specifically,
\begin{equation}
  p_\lambda \mathcal{C}(\vec{x}_\lambda) =
  \begin{cases}
    2 \sqrt{x_\star x_{n_1}} (a_\lambda + \frac{1}{3} c_\lambda) & \text{if $a_\lambda \geq c_\lambda$}, \\
    2 \sqrt{x_\star x_{n_1}} \bigl(\sqrt{a_\lambda c_\lambda} + \frac{1}{3} a_\lambda^{3/2} c_\lambda^{-1/2} \bigr) & \text{if $c_\lambda > a_\lambda$}.
  \end{cases}
\end{equation}
In both intervals, $p_\lambda \mathcal{C}(\vec{x}_\lambda)$ is an increasing function of $c_\lambda$.  Therefore, like before, the total average change in $\mathcal{C}$ is maximized when $c_1$ and $c_2$ are both large, i.e., $c_1 + c_2 = 1$.  Thus, taking $a_1$ such that $a_1\geq c_1$ (otherwise take $a_2=1-a_1$ such that $a_2\geq c_2$), and noting that the function $\mathcal{C}$ from Eq.~\eqref{Eq:Concfunc} reduces to $\frac{8}{3} \sqrt{x_\star x_{n_1}}$ prior to measurement, we get
\begin{align}
  \overline{\Delta\mathcal{C}}
  &= 2 \sqrt{x_\star x_{n_1}}
     \left[ \Bigl( a_1 + \frac{1}{3} c_1 \Bigr) +
            \Bigl( \sqrt{(1-a_1)(1-c_1)} + \frac{1}{3} (1-a_1)^{3/2} (1-c_1)^{-1/2} \Bigr)
          - \frac{4}{3}
     \right] \notag\\
&\approx -\frac{1}{3} \sqrt{x_\star x_{n_1}}(a_1-c_1)^3\leq 0.
\end{align}
Equality holds only with a trivial measurement.
\end{proof}

\bibliographystyle{alphaurl}
\bibliography{QuantumBib}

\end{document}